\documentclass[a4paper,10pt]{article}
\usepackage[utf8]{inputenc}
\usepackage{amsmath,amsthm,amssymb,amsfonts}

\usepackage[top=2.5cm,bottom=2.5cm,left=2cm,right=2.5cm]{geometry}
\usepackage{hyperref}
\usepackage{enumerate}

\input{macros}
\newcommand{\email}[1]{\mbox{\href{mailto:#1}{#1}}}
\newcommand{\optqed}{}

\newtheorem{definition}{Definition}[section]

\newtheorem{lemma}[definition]{Lemma}
\newtheorem{proposition}[definition]{Proposition}
\newtheorem{theorem}[definition]{Theorem}
\newtheorem{example}[definition]{Example}
\newtheorem{remark}[definition]{Remark}

\numberwithin{equation}{section}

\title{Relative entropy of coherent states on general CCR algebras}
\author{%
Henning Bostelmann\thanks{%
University of York, Department of Mathematics, York YO10 5DD, United Kingdom; \newline e-mail: \email{henning.bostelmann@york.ac.uk}}
\and 
Daniela Cadamuro\thanks{%
Institut f\"ur Theoretische Physik, Universit\"at Leipzig, Br\"uderstra\ss{}e 16, 04103 Leipzig, Germany;  \newline e-mail: \email{cadamuro@itp.uni-leipzig.de} }
\and 
Simone Del Vecchio\thanks{%
Institut f\"ur Theoretische Physik, Universit\"at Leipzig, Br\"uderstra\ss{}e 16, 04103 Leipzig, Germany; \newline e-mail: \email{simone.del\_vecchio@physik.uni-leipzig.de} }
}

\date{December 6, 2021}

\begin{document}

\maketitle

%
%
%
%
%
%
%
%
%
%

\begin{abstract}
For a subalgebra of a generic CCR algebra, we consider the relative entropy between a general (not necessarily pure) quasifree state and a coherent excitation thereof. We give a unified formula for this entropy in terms of single-particle modular data. Further, we investigate changes of the relative entropy along subalgebras arising from an increasing family of symplectic subspaces; here convexity of the entropy (as usually considered for the Quantum Null Energy Condition) is replaced with lower estimates for the second derivative, composed of ``bulk terms'' and ``boundary terms''. Our main assumption is that the subspaces are in \emph{differential modular position}, a regularity condition that generalizes the usual notion of half-sided modular inclusions. We illustrate our results in relevant examples, including thermal states for the conformal $U(1)$-current.
\end{abstract}

\section{Introduction}

Entropy and related correlation measures are of fundamental importance in quantum physics; not only in information theory, but also in thermodynamics and quantum field theory.

Mathematically, the most appropriate generalization of the classical notion of (relative) entropy to quantum systems, or noncommutative probability spaces, is formulated in terms of normal states on von Neumann algebras \cite{Araki:relent1} (see also \cite{OP04,Berta:smoothentropy}). However, while the formalism is quite easy to handle for type~I factors, where normal states are described by positive trace-class operators and the entropy can be computed by means of traces, applications to the type~$\mathrm{III}_1$ factors occurring generically in quantum field theory \cite{BDF:universal_structure} require working with (relative) Tomita-Takesaki modular objects, which are difficult to decribe explicitly in examples.

Recent work in quantum field theory \cite{LongoXu:relentCFT,Longo:entropyDist} has focussed on entropy measures for algebras associated with certain subregions of spacetime, and the dependence of the entropy of a given state depending on the spacetime region. Specifically, one considers the relative entropy between a ground state and a coherent excitation in the setting of linear fields \cite{Casini:relent,CLR:waveinfo} or related situations in chiral conformal quantum field theories \cite{hollands2020relative,Panebianco:relent,Panebianco:loop}; in some geometric situations, specific information about the (relative) modular operator is available here and allows for explicit results. 

Let us illustrate the situation in an example, following \cite{CLR:waveinfo}. Consider a massive free field in 3+1-dimensional Minkowski space, given in terms of the well-known symplectic space $(\kcal,\sigma)$ and real subspaces $\lcal(\ocal)\subset\kcal$ associated with space-time regions $\ocal$, and the corresponding Weyl (CCR) algebras $\A(\ocal)$. Further let $\omega$ be the vacuum state on these algebras, and consider a \emph{coherent state} $\omega_g = \omega(W(g)^\ast \cdot W(g))$, where $g \in  \kcal$  and $W(g)$ is the corresponding Weyl operator. Consider the standard left wedge $\wcal = \{x: x^1 < 0, |x^0|<|x^1| \} \subset \rbb^4$ , and for $t\in \rbb$ the shifted region\footnote{%
Here and in the following, our conventions are arranged so that larger values of the parameter $t$ correspond to larger regions (and correspondingly, larger symplectic spaces, algebras, etc.); the literature often chooses the opposite sign.
} 
$\wcal_t = \wcal + (t,t,0,0)$. Then the relative entropy between $\omega_g$ and $\omega$ with respect to the algebra $\A(\wcal_t)$ can be computed as \cite{CLR:waveinfo}
\begin{equation}
\relent{\omega_g}{\omega}{\A(\wcal_t)} = 2\pi \int_{x^1  < t} d\xv  \;(t-x^1) \; T_g^{00}(t, \xv),
\end{equation}
where $T_g^{\mu\nu}$ is the single-particle stress-energy tensor of the wave function $g$. Consequently, with $v = (1,1,0,0)$,
\begin{align}
\label{eq:wedged1}
  \frac{d}{d t} \relent{\omega_g}{\omega}{\A(\wcal_t)} &= 2\pi \int_{x^1<t} d\xv  \; v_\mu T^{0\mu}_g(t, \xv)  \geq 0,
\\
\label{eq:wedged2}
\frac{d^2}{dt^2} \relent{\omega_g}{\omega}{\A(\wcal_t)} &= 2\pi \int_{x^1=t} d\xv  \; v_\mu v_\nu T^{\mu\nu}_g (t, \xv)  \geq 0.
\end{align}
The second derivative is nonnegative, and hence the relative entropy is a convex function of $t$; this can be regarded \cite{CeyhanFaulkner:recovering} as a variant of the Quantum Null Energy Condition (QNEC). More generally, the QNEC is understood as a relation between certain expectation value of the energy density and the second derivative of the relative entropy \cite{Bousso:proofQNEC}, which is also suggested by Eq.~\eqref{eq:wedged2}. In this paper, we will only investigate derivatives of the entropy along a family of regions or subspaces, but will not comment on the relation with the energy density.

Apart from convexity, one may observe that the first derivative \eqref{eq:wedged1} is given by a ``bulk term'' (an integral over a Cauchy surface for the wedge region) 
while the second derivative \eqref{eq:wedged2} is given by a ``boundary term'' (an integral over the edge of the wedge at $x^1=x^0=t$).

This motivates the question which of these observations are a coincidence of the specific system chosen, and which of them generalize to a wider context.

In this paper, we ask such questions in a generic setting. We remain within the context of CCR algebras, i.e., the algebras in question are still generated by the ``second quantization functor'' from a symplectic space $\kcal$ and certain real subspaces $\lcal\subset\kcal$; and our states will be of the quasifree type. However, we abstract from the specifics of the above example.

As a first point, we investigate the connection between the symplectic (single-particle) structure and the relative entropy on the CCR algebras. Essentially, the methods of \cite{CLR:waveinfo} apply whenever the symplectic subspace $\lcal\subset\kcal$ above is standard and factorial, and the state $\omega$ is quasifree and pure. (These notions will be recalled in Sec.~\ref{sec:nonpure}.) However, in applications in physics, also non-pure quasifree states are of importance, for example thermal states \cite{RST:KMS,BraRob:qsm2} or Hadamard states in quantum field theory on curved spacetimes \cite{KW91,Radzikowski:microlocal}. Moreover, while factorial subspaces are usual in quantum field theory, they are certainly not the most general case (cf.~\cite{Verch:adjoint}).  

We aim to prove a unified formula for the relative entropy between a quasifree state $\omega$ and an associated ``coherent excitation'' $\omega_g$ in the general case. Our approach is as follows. We start with a generic symplectic space and consider the CCR algebra over it, equipped with a quasifree state. The state is not assumed to be pure; rather, using the well-known purification construction \cite{Woronowicz:purification,KW91}, we extend it to a pure state on a larger algebra.  
Now given a closed subspace $\lcal$, we decompose the extended space (and the corresponding CCR algebra) into factorial, abelian and nonseparating parts, and compute the relative entropy for these. We give a unified formula for the relative entropy between coherent states with respect to $\A(\lcal)$, where $\lcal$ is a generic subspace, in terms of the modular data associated with $\lcal$.

Second, we consider a family of subspaces $\{\lcal_t\}$, depending on a real parameter $t$, in particular when $\lcal_t$ increases with $t$; we ask how the relative entropy $S_t(g)=\relent{\omega_g}{\omega}{\A(\lcal_t)}$ for given $g \in \kcal$ changes with $t$. 

To this end, the following technical insight is important. With each subspace $\lcal_s$ one obtains, as in \cite{CLR:waveinfo}, a projector $Q_s$ which projects onto the ``$\lcal_s$-entropy relevant part'' of $\kcal$ (in the factorial case, onto $\lcal_s$ itself) and annihilates the symplectic complement $\lcal_s'$. However, this projector is unbounded in the usual topology of the symplectic space $\kcal$; even more, its domain will usually depend on the parameter $s$, which makes it particularly challenging to analyze a change in the parameter. However, let us equip the space with an (indefinite) scalar product arising from the semipositive quadratic form $S_t(g)$, where $t \neq s$ in general. With respect to this Hilbert space structure, it turns out in relevant cases that the projector $Q_s$ is \emph{orthogonal}, in particular bounded. We say in this case that the spaces $\lcal_s$, $\lcal_t$ are in \emph{differential modular position}, a condition that underlies our analysis, and resembles the concept of geometric modular action (see \cite{Bor:ttt}). 

This structure then allows for the desired analysis of bulk vs.~boundary terms: For fixed $g \in \kcal$, let us consider the function $T_g(s,t) = S_t(Q_s g)$, which equals the entropy on the diagonal $t=s$. A change in $s$ near the diagonal then corresponds to an abstract ``boundary change'' while a change in $t$ corresponds to a ``bulk change''. Analyzing the monotonicity properties of $T_g$, we establish estimates between the partial derivatives of $T_g$ (at $s=t-0$) and the desired derivatives of the entropy. 

We note that convexity of $S_t$ cannot be expected in such a general setting; it is already not preserved under a smooth reparametrization of the family of spaces, which our definition admits. However, we establish lower estimates on the second derivative (in the sense of distributions) that replace convexity.
Also, the observation from above that the first derivative contains only bulk terms, while the second derivative contains only boundary terms, does not hold up in general, and is replaced by a more nuanced picture. 

We verify the regularity condition of ``differential modular position'' in a number of examples, mainly but not exclusively from quantum field theory. In particular, it turns out that in half-sided modular inclusions of symplectic subspaces (cf.~\cite{Bor:ttt,LonNotes}), our condition is always fulfilled. Also, we treat relative entropies for halfline algebras in the conformal $U(1)$-current in thermal states, which to our knowledge have not appeared in the literature.

The paper is organized as follows: Sec.~\ref{sec:nonpure} defines our setting, recalls the purification and decomposition construction for symplectic spaces, and establishes the unified formula for relative entropies in terms of single-particle objects. In Sec.~\ref{sec:subspaces}, we investigate the relative positions of several subspaces, in particular one-parameter families of inclusions; we formulate our main condition (differential modular position) and derive estimates for the second derivative of the relative entropy. Then we show that all half-sided modular inclusions fit into our framework (Sec.~\ref{sec:halfsided}). In Sec.~\ref{sec:examples} we give examples from quantum mechanics, quantum field theory and classical probability theory in which our framework is applicable, illustrating various cases that can occur with respect to the derivative estimates we established. We end with a conclusion and outlook in Sec.~\ref{sec:conclusion}. 
The appendix recalls definition and fundamental properties of the relative entropy on $C^\ast$- and von Neumann algebras.

\section{Entropies in nonpure states}\label{sec:nonpure}

We first introduce our setting of symplectic spaces, purification and the decomposition of subspaces in Sec.~\ref{sec:ops}. 
Then (Sec.~\ref{sec:ccrentro}) we pass to the associated CCR algebras and their decomposition, and express the relative entropy between coherent states in terms of the single-particle modular data. Sec.~\ref{sec:approx} establishes some approximation properties needed in later sections.

\subsection{Single-particle structure}\label{sec:ops}
The basic object we work with is as follows:
\begin{definition}\label{def:space}
   Let $\kcal$ be a vector space over $\rbb$ and $\tau,\sigma$ two bilinear forms on $\kcal$. The triple $(\kcal,\tau,\sigma)$ is called a \emph{symplectic Hilbert space} if $(\kcal,\tau)$ is a separable Hilbert space, $(\kcal,\sigma)$ is a symplectic space, and if
   \begin{equation}\label{eq:sigmabd}
     \forall f,g \in \kcal: \quad \sigma(f,g)^2 \leq \tau(f,f) \tau(g,g).
   \end{equation}
\end{definition}

Here $\sigma$ is allowed to be degenerate as a symplectic form; we can (and will) assume without loss of generality that the dimension of its kernel is either even or infinite. (Otherwise consider the direct sum of $\kcal$ with a one-dimensional space, on which $\sigma$ is set to vanish.) Note that $\kcal$ is assumed a priori to be complete with respect to $\tau$-convergence; in applications one often starts with a pre-Hilbert space in the first step, and then takes its completion, but note that e.g.\ a non-degenerate form $\sigma$ on the non-completed space might be degenerate on the completion (cf.~\cite{Verch:adjoint}).

If $\kcal$ is a Hilbert space over $\cbb$ with complex scalar product $\hscalar{\cdotarg}{\cdotarg}$, then a standard example for Definition~\ref{def:space} is $\tau(f,g)=\re \hscalar{f}{g}$ and $\sigma(f,g)=\im \hscalar{f}{g}$. In fact, this is exactly the case when the quasifree state induced by $\tau$ on the CCR algebra over $(\kcal,\sigma)$ (see Sec.~\ref{sec:ccrentro} below) is a pure state \cite{MV:quasifree}; hence we will call $(\kcal,\tau,\sigma)$ \emph{pure} in this case. In general, it is always possible to embed $(\kcal, \tau, \sigma)$ into a pure symplectic Hilbert space $(\kcal\pf, \tau\pf, \sigma\pf)$, i.e., such that $\tau\pf$, $\sigma\pf$ are extensions of $\tau$, $\sigma$. This construction is known as \emph{purification}, and we will present it here in the form of \cite[Ch.~4]{Petz:ccr}; see also \cite[Appendix~A]{KW91}.

Due to \eqref{eq:sigmabd}, we can write $\sigma = \tau(\cdotarg, D \cdotarg)$ with an operator $D$, where $\|D\|\leq 1$. Using the polar decomposition of $D$ on the orthogonal complement of $\ker D$, and a suitable choice\footnote{%
At this point, our assumption enters that the dimension of $\ker D$ is either even or infinite.}
on $\ker D$, we obtain two bounded operators $C$, $|D|$ such that
\begin{equation}
   C|D| = D = |D| C, \quad |D| \geq 0, \quad C^2 = -1, \quad D^\dagger = -D,  \quad C^\dagger = -C.
\end{equation}
(We denote the adjoint with respect to $\tau$ by $\dagger$, whereas we will denote adjoints on complex Hilbert spaces by $\ast$ later on.)

We now define the space $\kcal\pf:=\kcal \oplus \kcal$, which is a vector space over $\cbb$ with respect to the complex structure given by the operator
\begin{equation}\label{complexstr}
   \ipf = \begin{pmatrix}
             -D  & C\sqrt{1+D^2}
             \\
             C \sqrt{1+D^{2}} & D
         \end{pmatrix}.
\end{equation}
In fact, defining the bilinear forms ($f,g \in \kcal\pf$)
\begin{align}
   \tau\pf &:= \tau \oplus \tau, 
   \\
   \sigma\pf (f, g) &:=  \tau\pf(f, -\ipf g),
   \\
   \hscalar{f}{g}\pf &:= \tau\pf (f, g) + i  \sigma\pf (f, g) ,
\end{align}
$\kcal$ becomes a complex Hilbert space with the scalar product $\hscalar{\cdotarg}{\cdotarg}\pf$, and a nondegenerate symplectic space with symplectic form $\sigma\pf$.
Identifying $\kcal$ with $\kcal\oplus 0$, the restrictions of $\tau\pf$ and $\sigma\pf$ to $\kcal \times \kcal$ are $\tau$ and $\sigma$ respectively, as the notation suggests.

Now let $\lcal\subset \kcal$ be a closed subspace. (Note that closure in $\kcal$-norm is crucial for the following.) 
We decompose $\lcal$ in a standard way (cf.~\cite{Halmos:subspaces}) as follows: We set
\begin{align}
   \label{eq:l0def}
   \lcal_0\pf & := (\lcal + \ipf \lcal)^\perp,  \\
   \lcal_\infty\pf & := \lcal \cap \ipf \lcal \equiv \lcal_\infty, \\
   \lcal_{\abel} & := \lcal \cap \lcal', \\
   \lcal_{\abel}\pf & := \lcal_{\abel}  + \ipf \lcal_{\abel} , \\
   \lcal_{\fact}\pf & := ( \lcal_0\pf \oplus \lcal_{\abel}\pf \oplus \lcal_\infty\pf)^\perp , \\
   \label{eq:lfdef}
   \lcal_{\fact} & := \lcal_{\fact}\pf \cap \lcal,
\end{align}
where $\lcal'$ denotes the symplectic complement of $\lcal$.
The spaces $\lcal_\fact$, $\lcal_\abel$, and $\lcal_\infty$ are called the factorial, abelian and nonseparating parts of $\lcal$, respectively, for reasons that will become clear below. We then have:
\begin{lemma}\label{lemma:kdecomp}
   $\kcal\pf$ is isomorphic to the orthogonal direct sum
   \begin{equation}
      \kcal\pf \isom \lcal_0\pf \oplus \lcal_{\abel}\pf \oplus \lcal_{\fact}\pf \oplus \lcal_\infty\pf  
   \end{equation}
   and under this isomorphism
   \begin{equation}
      \lcal \isom  0 \oplus \lcal_{\abel} \oplus \lcal_{\fact} \oplus \lcal_\infty.  
   \end{equation}
\end{lemma}
\begin{proof}
  One shows by direct computation that $\lcal_{\abel}$ is complex-orthogonal to $\lcal_\infty$; also, $\lcal_\abel$ is real-orthogonal to $\ipf\lcal_\abel$, hence $\lcal_\abel\pf$ is closed. The other parts follow directly from the definitions \eqref{eq:l0def}--\eqref{eq:lfdef}.
\optqed{}\end{proof}

All three components of $\lcal$ may be present in general: in quantum field theory, one usually considers purely factorial subspaces, i.e., $\lcal=\lcal_{\fact}$ (see Examples~\ref{ex:u1vac} and \ref{ex:u1kms}); but in other situations, $\lcal$ may be purely abelian ($\lcal=\lcal_{\abel}$, Example~\ref{ex:abelian}), or one may have $\lcal=\lcal_\infty$ (part of Example~\ref{ex:fdstraight}), and of course direct sums of these can be formed. We note some special cases:

\begin{remark}\label{remark:pure}
   If $(\kcal,\tau,\sigma)$ is a \emph{pure} symplectic Hilbert space, then $D=-i$, hence $\ipf$ acts by the diagonal matrix $\left(\begin{smallmatrix}i & 0 \\ 0 & -i\end{smallmatrix}\right)$. In the decomposition of Lemma~\ref{lemma:kdecomp}, this leads to $0 \oplus \kcal \subset \lcal_0\pf$, and all other spaces $\lcal_\fact$,  $\lcal_\abel$, $\lcal_\infty$ etc.\ being contained in $\kcal \oplus 0$. In this sense, if $(\kcal,\tau,\sigma)$ is already pure, we can ignore the purification construction.
\end{remark}
\begin{remark}\label{remark:infinity}
   If specifically $\kcal=\lcal$ in Remark~\ref{remark:pure}, then $\lcal_\infty=\kcal\oplus 0$, $\lcal_\fact=\lcal_\abel = \{ 0\}$, $\lcal_0\pf=0 \oplus \kcal$. 
\end{remark}
\begin{remark}\label{remark:abelian}
   For a symplectic Hilbert space $(\kcal,\tau,0)$ (i.e., for $\sigma=0$), even- or infinite-dimensional, we obtain $i\pf = \left(\begin{smallmatrix} 0 & C \\ C & 0  \end{smallmatrix}\right)$. The map $(f,g) \mapsto f - i Cg$ then identifies $\kcal\pf$ with the usual complexification of $\kcal$. For $\lcal\subset\kcal$, we have $\lcal=\lcal_\abel$, $\lcal_\fact=\lcal_\infty=\{0\}$, $\lcal_0\pf = \lcal^\perp + i \lcal^\perp$ where $\perp$ denotes the orthogonal complement in~$\kcal$.
\end{remark}

In the following, we shall denote the complex-linear orthogonal projectors onto $\lcal_\fact\pf$ etc.\ as $P_\fact\pf$ etc. We also denote by $P_\abel$ the real-orthogonal projector onto $\lcal_\abel$, and by $P_\fact$ the real-linear projector with image $\lcal_\fact$ and kernel $\lcal_\fact '$. Note that $P_\fact$ is not bounded (or orthogonal) in general, but closed on its domain $\lcal_\fact + \lcal_\fact'$ \cite{CLR:waveinfo}.

We also consider the subspaces $\lcal_\std := \lcal_\abel \oplus \lcal_\fact$,  $\lcal_\std\pf:=\lcal_{\abel}\pf \oplus \lcal_{\fact}\pf$; here $\lcal_\std \subset \lcal_\std\pf$ is \emph{standard} in the sense that $\lcal_\std \cap \ipf \lcal_\std= \{0\}$ and $\lcal_\std + \ipf \lcal_\std$ is dense in $\lcal_\std\pf$. Hence \cite{RieffelvanDaele:boundedttt} we obtain Tomita-Takesaki objects $J_\lcal$, $\Delta_\lcal$ with respect to this subspace. We set $K_\lcal := -\log \Delta_\lcal$, then extend this operator $K_\lcal$ by 0 to $\lcal_0\pf$ and consider it as undefined on $\lcal_\infty\pf \backslash \{0\}$. We denote the modular group by $U_\lcal(x)=\exp( - \ipf x K_\lcal)$, again defined on $\lcal_0\pf \oplus \lcal_\std\pf$.
It is important in the following that the projector $P_\fact$ can be written as a function of the modular objects:
\begin{lemma} (\cite[Theorem~2.2]{CLR:waveinfo}) \label{lemma:pfmodular}
 Let $a(\lambda) = (1-\lambda)^{-1}$, $b(\lambda) = \lambda^{1/2}a(\lambda)$. Then
 \begin{equation}
P_\fact = \big( a(\Delta_\fact ) + J_\lcal b(\Delta_\fact)   \big)^- \quad \text{where } \Delta_\fact = \Delta_\lcal \restriction \lcal_\fact\pf.  
 \end{equation}
\end{lemma}

For use in future sections, we also consider the closed, real-linear projector 
\begin{equation}\label{def:cutproj}
 \begin{aligned}
  Q_\lcal &= 0 \oplus (1-P_{\abel}) \oplus P_{\fact}  \oplus \idop 
  \\
\text{with domain  }\;\dom Q_\lcal &= \lcal_0\pf \oplus \lcal_{\abel}\pf \oplus (\lcal_\fact + \lcal_\fact') \oplus    \lcal_\infty\pf .
\end{aligned}
\end{equation}
Note that $\img Q_\lcal = \lcal$ in the purely factorial case ($\lcal=\lcal_\fact$), but in general $\operatorname{img} Q_\lcal \neq \lcal$; rather, as will become clear in the next subsection, $Q_\lcal$ projects onto the ``entropy-relevant part'' of the space. (See Lemma~\ref{lem:entrform}(\ref{it:sq}) and Theorem~\ref{theorem:entroGen} in particular.) However, we always have $\operatorname{ker} Q_\lcal = \lcal'$. In other words, $\operatorname{img}Q_\lcal \neq (\operatorname{ker}Q_\lcal)'$ in general. We also note:

\begin{lemma}\label{lemma:spectralcut}
   For $0<\epsilon <1$, let $Q^{(\epsilon)}$ be the spectral projector of $\log \Delta_\lcal$ for the set $(-\epsilon^{-1}, -\epsilon) \cup (\epsilon, \epsilon^{-1}) \cup \{0\}$, extended by $\idop$ to $\lcal_0\pf$ and $\lcal_\infty\pf$.  
   Let $\dcut := \cup_{0<\epsilon<1} Q^{(\epsilon)} \kcal\pf$. Then $\dcut$ is a core for $Q_\lcal$, and $\dcut\cap \lcal_\std\pf$ a common core for $\Delta_\lcal$ and  $\log \Delta_\lcal$.
\end{lemma}
\begin{proof}
We can suppose without loss of generality that we are in the factorial case, i.e., $\lcal=\lcal_{\fact}$, since on $\lcal_\abel^{\oplus}$ we have that $\Delta_\lcal\restriction {\lcal_\abel^{\oplus}}=1$, $\log \Delta_\lcal\restriction{\lcal_\abel^{\oplus}}=0$, and $Q_\lcal\restriction{\lcal_\abel^{\oplus}}$ is bounded, while on $\lcal_0^\oplus$ and $\lcal_\infty^\oplus$ the statement is clearly trivial.  That $\dcut\cap \lcal_\std\pf$ is a common core for $\Delta_{\lcal}$ and $\log \Delta_{\lcal}$ is immediate by functional calculus. That $\dcut$ is a core for $Q_\lcal$ in the factorial case follows by the expression of $Q_\lcal$ in terms of $\Delta_\lcal$ and $J_\lcal$ given in Lemma~\ref{lemma:pfmodular}.
\optqed{}\end{proof}

\subsection{CCR algebras and relative entropy}\label{sec:ccrentro}

We now pass to the CCR algebras on the symplectic space $(\kcal,\sigma)$; see, e.g., the monographs \cite{Petz:ccr,DerezinskiGerard:quantization}. We denote by $\A_\kcal:=\ccr(\kcal,\sigma)$ the $C^\ast$ algebra generated by elements $W(f)$, $f \in \kcal$, with the relations
\begin{equation}\label{eq:weyl}
   W(f)W(g) = e^{-i\sigma(f,g)}W(f+g), \quad W(f)^\ast = W(-f).
\end{equation}
Similarly, for a closed subspace $\lcal\subset\kcal$, we define $\A_{\lcal} :=\ccr(\lcal,\sigma)\subset \A_\kcal$, $\A_{\kcal}\pf := \ccr(\kcal\pf,\sigma\pf)\supset \A_\kcal$, and write the relevant subalgebras as $\A_\infty := \ccr(\lcal_\infty,\sigma\pf)$ etc. 

On $\A_\kcal$, the bilinear form $\tau$ induces the quasifree state\footnote{also known in the literature as a quasifree state \emph{with vanishing one-point function}} $\omega$ by
\begin{equation}
   \omega(W(f)) = e^{-\tau(f,f)/2};
\end{equation}
we use the same notation for its extension by $\tau\pf$ to $\A_\kcal\pf$ and the restrictions to subalgebras, suppressing the dependence on $\tau$ where no confusion can arise. Related to $\omega$, for each $g \in \kcal$ we consider the coherent state\footnote{An alternative nomenclature is \emph{quasifree state with nonvanishing one-point function}.}
\begin{equation}
   \omega_g = \omega(W(g)^\ast \cdotarg W(g));
\end{equation}
note that $\omega_0=\omega$.

We are interested in the relative entropy between the $\omega_g$ (for different $g$) as states on the $C^\ast$-algebra $\A_\lcal$; see Appendix~\ref{app:entro} for a brief review of this concept. As a first step, we remark that the relative entropy respects the decomposition of $\lcal$:

\begin{proposition}\label{prop:entrosum}
  Let $(\kcal,\tau,\sigma)$ be a symplectic Hilbert space. For any closed subspace $\lcal\subset\kcal$, we have
  \begin{equation}\label{eq:entrosum}
    \relent{\omega_g}{\omega}{ \A_\lcal} 
    = 
    \relent{\omega_{P_{\abel} \pf g}}{\omega}{ \A_{\abel}} 
+
    \relent{\omega_{P_{\fact} \pf g}}{\omega}{ \A_{\fact}} 
+
    \relent{\omega_{P_\infty\pf g}}{\omega}{ \A_\infty} .
\end{equation}
\end{proposition}

\begin{proof}
Due to Lemma~\ref{lemma:kdecomp}, and noting that the pure quasifree states are faithful on the respective subalgebras, we know that $\A_{\kcal}\pf$ is isomorphic to the (spatial) tensor product of $C^\ast$-algebras
\begin{equation}
     \A_{\kcal}\pf \isom \A_0\pf \otimes \A_{\abel}\pf \otimes \A_{\fact}\pf \otimes \A_\infty\pf  
\end{equation}
and under this isomorphism
\begin{equation}
    \A_\lcal \isom  \cbb\idop \otimes \A_{\abel} \otimes \A_{\fact} \otimes \A_{\infty}
\end{equation}
and
\begin{equation}
   \omega_{g} \isom \omega_{P_0\pf g} \otimes \omega_{P_{\abel}\pf g} \otimes \omega_{P_{\fact}\pf g} \otimes \omega_{P_\infty\pf g};
\end{equation} 
$\omega$ decomposes in the same way.
This decomposition holds analogously for the induced von Neumann algebras in the GNS representation of $\A_{\kcal}\pf$ associated with $\omega$. Thus, due to additivity of the relative entropy in this situation (see Lemma~\ref{lemma:additivity} in the appendix), we obtain \eqref{eq:entrosum}. (This includes the obvious observation that the summand with respect to $\lcal_0\pf$ vanishes.) 
\optqed{}\end{proof}

We will now compute the three terms in \eqref{eq:entrosum} individually. We start with the abelian part, following standard methods (cf.~\cite{VS:malliavin}).

\begin{proposition}\label{prop:entroA} For any $g \in \lcal_{\abel}\pf$,
   \begin{equation}\label{eq:relenta}    
        \relent{\omega_{g}}{\omega}{ \A_{\abel}} = 2 (\|  (1-P_{\abel}) g \|\pf)^2 
   \end{equation}
   where $P_{\abel}$ is the (real-linear) projector onto $\lcal_{\abel}$.
\end{proposition}

\begin{proof}
 Since $\kcal$ is separable, the von Neumann envelope of $\A_\lcal$ is generated by the algebras for finite-dimensional subspaces of $\lcal$.  Lemma~\ref{lemma:entrounion} in the appendix shows that $\relent{\omega_{g}}{\omega}{ \A_{\abel}}$ is determined by the supremum of the entropy for these subalgebras; hence it suffices to prove the statement for the case of finite-dimensional $\lcal_a\pf$. Also, on the algebra $\A_\abel$, the state $\omega_{g}$ coincides with $\omega_{\hat g}$ where $\hat g = (1 - P_a)g$; hence we can assume without loss that $g \in (1-P_a)\lcal_a\pf = \ipf \lcal_a$.
 
 In this case, after a suitable choice of basis, $\lcal_\abel\pf$ with the scalar product $\hscalar{\cdot}{\cdot}\pf$ can be identified with $\cbb^n$ and its  standard scalar product, with the real subspace $\rbb^n$ corresponding to $\lcal_a$. The GNS representation $\pi$ for $(\A_{a}, \omega)$ acts on $L^2(\rbb^n,d\mu)$ where $d\mu = (2\pi)^{-n/2}\exp(-\|x\|^2/2)d^nx$, with $\pi( W(f) )$ being multiplication with $\exp i \hscalar{f}{ \cdotarg}$, and  $\pi( \A_{a} )'' = L^\infty(\rbb^n,d\mu)$. The states $\omega$ and $\omega_g$ are vector states with vectors 
 $\Omega(x) = 1$, $\Omega_g( x) = \exp (\hscalar{\ipf g}{x} - (\|g\|\pf)^2)$. 
 The relative modular group turns out to act by multiplication with $ \exp ( -2it\hscalar{\ipf g}{x} + 2it(\|g\|\pf)^2)$.  
The relative entropy can then be computed from the general definition \eqref{eq:entropydef}, which yields the result \eqref{eq:relenta}.
\optqed{}\end{proof}

Of course, this relative entropy coincides with the usual Kullback-Leibler divergence of Gaussian distributions (cf.~\cite[p.~81]{OP04}). In the proof, we have used our simplifying assumption that $\kcal$ is separable, but by methods of the theory of Gaussian fields \cite{VS:malliavin}, we expect that this assumption is actually dispensable.

Next, we consider the factorial part, for which the relative entropy is known from \cite{CLR:waveinfo}.

\begin{proposition}\label{prop:entroF}
For any $g \in \lcal_{\fact}\pf\cap \dom K_\lcal$, one has $\ipf K_\lcal g \in \dom P_\fact$ and
   \begin{equation}\label{entrform}
        \relent{\omega_{g}}{\omega}{ \A_\fact} = \sigma\pf( g, P_\fact \, \ipf K_\lcal g).
   \end{equation}
\end{proposition}

\begin{proof}
We sketch the relevant techniques from \cite{CLR:waveinfo}. Since $(\lcal_\fact\pf,\tau\pf,\sigma\pf)$ is pure, the GNS representation $\pi$ of $(\A_\fact, \omega)$ acts on the Fock space over $\lcal_\fact\pf$, and in that representation both $\omega$ and $\omega_g$ are vector states: $\omega$ corresponds to the Fock vacuum vector $\Omega$, and $\omega_g$ to the vector $\Omega_g:=\pi(W(g)) \Omega$. The vector $\Omega$ is cyclic and separating for $\pi(\A_\fact)''$, and the associated Tomita-Takesaki modular group is $\Delta_{\Omega}^{it}= \Gamma( \Delta_\lcal^{it} )$, the ``second quantization'' of the unitary $\Delta_\lcal^{it} \restriction \lcal_\fact\pf$. 

Now first let $g \in \lcal_\fact \cap \dom K_\lcal$. Using that $W(g) \in \A_\fact$, one finds 
$\Delta_{\Omega,\Omega_g}^{it}= \Delta_\Omega^{it}$,
and consequently
\begin{equation}
\begin{aligned}
\relent{\omega_g}{\omega}{\A_\fact}
&= i \frac{d}{dt} \hscalar{\Omega_g}{\Delta_{\Omega,\Omega_g}^{it} \Omega_g} \Big\vert_{t=0}
= i \frac{d}{dt} \hscalar{\Omega}{\pi(W(g))^*\Delta_\Omega^{it}\,\pi(W(g))\Omega} \Big\vert_{t=0}\\
&=i \frac{d}{dt} \hscalar{\Omega}{\pi(W(g))^*\Delta_\Omega^{it}\,\pi(W(g))\Delta_\Omega^{-it}\Omega} \Big\vert_{t=0}.
\end{aligned}
\end{equation}
With the Weyl relations \eqref{eq:weyl} and $\Delta_\Omega^{it}=\Gamma( \Delta_\lcal^{it}) $, 

\begin{equation}
\pi(W(g))^*\Delta_\Omega^{it}\,\pi(W(g))\Delta_\Omega^{-it}
= \pi(W(g))^\ast \pi(W(\Delta_\lcal^{it} g))
= \pi(W(\Delta_\lcal^{it} g - g)) e^{ i\sigma\pf(g, \Delta_\lcal^{it} g ) }.
\end{equation}
Therefore the relative entropy is 
\begin{equation}
\relent{\omega_g}{\omega}{\A_\fact}
= i \frac{d}{dt} e^{-(\| \Delta_\lcal^{it} g-g \|\pf)^2/2} e^{ i\sigma\pf(g, \Delta_\lcal^{it} g ) }\Big\vert_{t=0}
=  \sigma\pf(g, \ipf K_\lcal g).
\end{equation}
Hence \eqref{entrform} holds for $g \in \lcal_f \cap \dom K_\lcal$. It also holds for $g \in \lcal_\fact' \cap \dom K_\lcal$, since in that case both 
sides of the equation vanish. The result for general $g \in \lcal_\fact\pf \cap \dom K_\lcal$ follows by a density argument that employs Lemma~\ref{lemma:pfmodular}; see \cite[Sec.~4.4]{CLR:waveinfo}. 
\optqed{}\end{proof}

On the nonseparating part, one finds the relative entropy as follows:

\begin{proposition}\label{prop:entroInf}
For any $g \in \lcal_\infty\pf$,
   \begin{equation}    
        \relent{\omega_{g}}{\omega}{ \A_{\infty}} = \begin{cases}
            0 \quad &\text{if $g = 0$,}
\\
            \infty \quad &\text{otherwise.}
\end{cases}
   \end{equation}
\end{proposition}

\begin{proof}
    Since $(\lcal_\infty,\tau\pf,\sigma\pf)$ is pure, the GNS representation $\pi$ of $(\A_{\infty},\omega$) is irreducible \cite[Ch.~4]{Petz:ccr} and $\omega$ and $\omega_g$ are given by vector states $\Omega$ and $\Psi:=\pi(W(g)) \Omega$ there. The support projections of these states are hence the projectors $P_\Omega$ and $P_\Psi$ respectively; and $P_\Omega \leq P_\Psi$ if and only if they are equal, i.e., for $g = 0$. The statement then follows from the definition of the relative entropy, see the appendix.
\optqed{}\end{proof}

Our goal is now to establish a unified formula that applies to all these cases, linking the relative entropy on the CCR algebras to a quadratic form at single-particle level. To that end:

\begin{lemma}\label{lem:entrform} 
(cf.~\cite[Prop.~2.5]{CLR:waveinfo})
Consider the real-linear operator on $\dcut\cap \lcal_\std\pf$,
\begin{equation}\label{eq:sop}
   R_\lcal := c(K_\lcal) (1-J_\lcal) c(K_\lcal), \quad \text{where } c(\lambda)=\sqrt{\frac{\lambda}{1-e^{-\lambda}} }\, ,
  \end{equation}
  extended by zero to $\lcal_0\pf$ and undefined on $\lcal_\infty\pf \backslash \{0\}$. (The function $c$ is extended by continuity to $\lambda=0$.)
  Then:
  \begin{enumerate}[(i)]
   \item \label{it:rpos} There is a unique closed real-linear quadratic form $S_\lcal$ associated with $R_\lcal$, which is positive;
   \item \label{it:rbd} one has $\|R_\lcal - c(K_\lcal)^2\|\leq 1$ as operators on $\lcal_0\pf \oplus \lcal_\std\pf$;
   \item \label{it:rdom} $\dom S_\lcal = \lcal_0\pf \oplus \lcal_\abel\pf \oplus \dom (E_+ |K_\lcal|^{1/2}) \oplus \{0\}$,
     where $E_+$ denotes the spectral projector of $K_\lcal$ for the interval $(0,\infty)$;

   \item \label{it:rker} $\ker S_\lcal = \lcal'$;
   \item \label{it:sq} $S_\lcal (Q_\lcal f, Q_\lcal f) = S_\lcal(f,f)$ for all $f \in \dom Q_\lcal \cap \dom S_\lcal$.
  \end{enumerate}  
\end{lemma}

\begin{proof}
Since $(1-J_\lcal)/2$ is a real-orthogonal projector, it is clear that $R_\lcal$ is positive. 
Thus $R_\lcal$ has a unique positive closed quadratic form associated with it, showing (\ref{it:rpos}). 
Further, one computes on $\lcal_\std\pf$,
\begin{equation}
  R_\lcal - c(K_\lcal)^2 = -J_\lcal  c(K_\lcal)c(-K_\lcal),
\end{equation}
and since $\lambda \mapsto c(\lambda)c(-\lambda)$ is bounded by 1, (\ref{it:rbd}) follows, also on $\lcal_0\pf$. Consequently, the form domain of $S_\lcal$ is the same as the operator domain of $c(K_\lcal)$; since $c(\lambda) \to 0$ as $\lambda \to -\infty$ and $c(\lambda) \sim \lambda^{1/2}$ as $\lambda \to \infty$, it can be written as in~(\ref{it:rdom}). We prove (\ref{it:rker}) separately for the restrictions to $\lcal_\abel\pf$ and $\lcal_\fact\pf$; it is trivial on $\lcal_0\pf$. Now on $\lcal_\abel\pf$, the statement follows from $c(0)=1$, while on $\lcal_\fact\pf$, one computes $\ker S_\lcal = \ker P_\fact$ by Lemma~\ref{lemma:pfmodular}, and $\ker P_\fact = \lcal' \cap \lcal_\fact\pf$.
Finally for (\ref{it:sq}), let $f \in \dom Q_\lcal$. Then $f=g+g'$ with $g \in \img Q_\lcal$ and $g'\in \ker Q_\lcal = \lcal' = \ker S_\lcal$. If additionally $f \in \dom S_\lcal$, then also $g \in \dom S_\lcal$, and $S_\lcal(f,f)=S_\lcal(g,g)=S_\lcal(Q_\lcal f,Q_\lcal f)$ follows.
\optqed{}\end{proof}

We will sometimes write $S_\lcal(f)$ as shorthand for $S_\lcal(f,f)$. We are now ready to state the main result of the section:
\begin{theorem}\label{theorem:entroGen}
For any $f,g\in\kcal\pf$, we have
  \begin{equation}\label{eq:entropylift}
     \relent{\omega_g}{\omega_{f}}{ \A_\lcal} = S_\lcal(g-f);    
  \end{equation}
  in particular, the left-hand side is finite if and only if $g-f \in \operatorname{dom} S_\lcal$.
\end{theorem}

\begin{proof}
  The automorphism $\alpha = \operatorname{ad} W(-f)$ of $\A_\lcal$ fulfills $\omega_{f} \circ \alpha = \omega$ and $\omega_g \circ \alpha = \omega_{g-f}$; hence we can assume $f=0$ without loss of generality.
  
  First let $g \in \lcal_\fact\pf \cap \dcut$. From Lemma~\ref{lemma:pfmodular}, spectral calculus shows  $P_\fact \ipf \log \Delta_\lcal = - \ipf R_\lcal$  on $\dcut\cap  \lcal_\fact\pf$,  thus
  \begin{equation}
     \sigma\pf(g, P_\fact \ipf \log \Delta_\lcal \,g) = \im \hscalar{g}{ P_\fact \ipf \log \Delta_\lcal \,g}\pf = \re \hscalar{g}{R_\lcal g}\pf =  \hscalar{g}{R_\lcal g}\pf,
  \end{equation}
  and \eqref{eq:entropylift} follows for all $g \in \lcal_\fact\pf \cap \dcut$ from Proposition~\ref{prop:entroF}. Using approximation techniques \cite[Theorem~4.5]{CLR:waveinfo}, the relation can be extended to all $g \in \lcal_\fact\pf$, including the case where the two sides of \eqref{eq:entropylift} are infinite.

Now consider $g \in \lcal_\abel\pf$. We note that $\lcal_\abel\pf = \ker\log \Delta_\lcal$ and $c(0)=1$, hence $R_\lcal$ acts as $1-J_\lcal= 2 (1-P_\abel)$ on $\lcal_\abel\pf$. Thus the proposed result \eqref{eq:entropylift} holds for $g \in \lcal_\abel\pf$, see Proposition~\ref{prop:entroA}.

Likewise, Proposition~\ref{prop:entroInf} shows that \eqref{eq:entropylift} holds for $g \in \lcal_\infty\pf$, with both sides being infinite unless $g=0$.
Applying  Proposition~\ref{prop:entrosum} now concludes the proof.
\optqed{}\end{proof}

\subsection{Approximation properties} \label{sec:approx}

For the following, we establish some approximation properties for the entropy form and the modular group.
Apart from the Hilbert space norm given by $\tau$ on $\kcal$ (and extended to $\gnorm{\cdotarg}{}\pf$ on $\kcal\pf$), we consider the following norms on $\kcal\pf$ or subsets of it:
\begin{itemize}
 \item the $K_\lcal$-graph norm, $\gnorm{f}{K,\lcal} := \|K_\lcal f\|\pf + \|f\|\pf$,
 \item the $S_\lcal$-graph norm, $\gnorm{f}{S,\lcal}^2 := S_\lcal(f) + \tau\pf(f,f)$,
 \item the seminorm $\|\cdot \|_\lcal$ defined by $\|f \|_\lcal^2:=S_\lcal(f)$.
\end{itemize}
It is clear that the $K_\lcal$-graph norm is stronger than the $S_\lcal$-graph norm, which is in turn stronger than the seminorm $\|\cdot \|_\lcal$.
We denote the closure of $\dom S_\lcal$ in $\|\cdot \|_\lcal$, modulo the kernel $\lcal'$ of the seminorm, as $\es_\lcal$; for formal reasons we explicitly denote the isometric inclusion of $(\dom S_\lcal, S_\lcal)$ into $\es_\lcal$ as $\varphi_\lcal$. Then $\es_\lcal$ becomes a Hilbert space with the (continuous extension of) the scalar product $\hscalar{\varphi_\lcal f}{\varphi_\lcal g}_\lcal=S_\lcal(f,g)$.

\begin{lemma}\label{lemma:graphcont}
 The modular group $U_\lcal$ maps $\dom S_\lcal$ into $\dom S_\lcal$, and this action is strongly continuous in the $S_\lcal$-graph norm. 
\end{lemma}

 \begin{proof}
 From Lemma~\ref{lem:entrform}(\ref{it:rdom}), it is clear that $U_\lcal=\exp( -\ipf K_\lcal)$ preserves $\dom S_\lcal$. Further, for $f \in \dom S_\lcal$ and $f_s := U_\lcal(s)f$, we have
 \begin{equation}
 \begin{aligned}
    S_\lcal(f_s-f,f_s-f) &\leq \hscalar{f_s-f}{c(K_\lcal)^2 (f_s-f)}\pf + ( \gnorm{f_s-f}{} \pf) ^2 
 \\   &= 2 \re \bighscalar{c(K_\lcal) f}{(\idop - U_\lcal(s)) \,c(K_\lcal) f}\pf   + ( \gnorm{f_s-f}{} \pf) ^2,
 \end{aligned}
 \end{equation}
where Lemma~\ref{lem:entrform}(\ref{it:rbd}) was used. This vanishes as $s \to 0$ due to strong continuity of $U_\lcal$ in the $\kcal\pf$-norm.
\optqed{}\end{proof}

The following lemmas for a fixed closed subspace $\lcal\subset\kcal$ will allow us to identify $\es_\lcal$ with a ``concrete'' Hilbert space in examples.

\begin{lemma}\label{density64}
$\varphi_\lcal(\img Q_\lcal \cap\dom K_\lcal)$ is dense in $\es_\lcal$ with respect to $\|\cdot\|_\lcal$. In particular, if $\lcal=\lcal_{\fact}$, then $\varphi_\lcal(\lcal \cap\dom K_\lcal)$ is dense in $\es_\lcal$ with respect to $\|\cdot\|_\lcal$.
\end{lemma}

\begin{proof}
Let $\varepsilon>0$ and $Q^{(\varepsilon)}$ be as in Lemma \ref{lemma:spectralcut}. Then for any $v\in\kcal\pf$, by Lemma~\ref{lemma:spectralcut} we have $Q^{(\varepsilon)}v\in  \dom Q_\lcal$. 
Also, for any $v\in\dom(S_\lcal)$, by the expression for the relative entropy in Theorem \ref{theorem:entroGen}, we have $Q^{(\varepsilon)} v\rightarrow v$ in $S_\lcal$-graph norm as $\varepsilon \rightarrow 0$. 
Furthermore for $v\in\dom(S_{\lcal})$, by functional calculus, $Q_\lcal Q^{(\varepsilon)}v\in  \img Q_\lcal\cap \dom(K_\lcal)$. Thus
\begin{equation}
\|\varphi_\lcal(v-Q_\lcal Q^{(\varepsilon)}v)\|_\lcal \leq S_\lcal\big(v-Q^{(\varepsilon)}v \big) + S_\lcal\big( Q^{(\varepsilon)}v-Q_\lcal Q^{(\varepsilon)}v \big)
= S_\lcal( v-Q^{(\varepsilon)}v )\rightarrow 0
\end{equation}
as $\varepsilon\rightarrow 0$.---%
If here $\lcal=\lcal_{\fact}$, one has $\img Q_\lcal=\lcal$ and the second statement follows.
\optqed{}\end{proof}

\begin{lemma}\label{core}
Let $\dcal\subset \lcal$ be a core for the generator of the strongly continuous (with respect to the norm of $\kcal$) one-parameter group $s\rightarrow U_\lcal(s) \restriction \lcal$, i.e., $\ipf K_\lcal \restriction \lcal$. Then $\dcal$ is dense in $\lcal \cap\dom K_\lcal$ in the $S_\lcal$-graph norm.
\end{lemma}

\begin{proof}
In this proof, we drop the index $\lcal$ on $K_\lcal$. Note that by functional calculus, the norm induced by the positive self-adjoint operator $c(K)^2+c(-K)^2$, where $\lambda\rightarrow c(\lambda)$ is as in Lemma \ref{lem:entrform},  is equivalent to the norm induced by $1+|K|$, while the graph norm of $\ipf K \restriction \lcal$ is induced by $1+K^2$ and thus stronger. Hence since $\dcal$ is a core for $\ipf K\restriction{\lcal}$, it is also dense in $\lcal \cap \dom K$ in the norm induced by $c(K)^2+c(-K)^2$ and consequently, by Lemma~\ref{lem:entrform}(\ref{it:rpos})--(\ref{it:rbd}), dense in the $S_\lcal$-graph norm.
\optqed{}\end{proof}

\section{Entropies for subspaces} \label{sec:subspaces}

We will now consider several subspaces $\lcal_t \subset \kcal$, and relations between the entropies related to them. To simplify notation, we will usually denote the related objects as $S_t$, $\Delta_t$, etc.\ rather than $S_{\lcal_t}$, $\Delta_{\lcal_t}$ etc.

\subsection{Two subspaces}

We begin with the relation between two subspaces, say, $\lcal_0$ and $\lcal_1$, and their associated entropy forms. Let us first mention:
\begin{lemma}\label{lemma:entropygrow}
   If $\lcal_0 \subset \lcal_1$, then $S_0(f,f) \leq S_1(f,f)$ for any $f \in \kcal$.
\end{lemma}
\begin{proof}
 This follows from Theorem~\ref{theorem:entroGen}, since the relative entropy is known to increase with the algebra considered (Lemma \ref{lem:restrentropy}).
\optqed{}\end{proof}

We now investigate the relation between the projector $Q_0$ (on the ``entropy relevant part'' for $S_0$) and the entropy $S_1$. Heuristically, we expect in relevant cases that
\begin{equation}\label{eq:sqortho}
   S_1( Q_0 f, g) = S_1(f, Q_0 g),
\end{equation}
in other words, that the projector $Q_0$ is ``orthogonal'' with respect to the bilinear form $S_1$. However, the relation \eqref{eq:sqortho} needs to be read with care, as in general neither domain nor image of $Q_0$ will consist only of vectors of finite entropy $S_1$. The precise version of our condition is given as:

\begin{definition}\label{def:dmp}
   Let $\lcal_0$ and $\lcal_1$ be closed subspaces of a symplectic Hilbert space ($\kcal$, $\tau$, $\sigma$). Define
   \begin{equation}\label{fundamdom}
      \dcal_{01}^+ :=  \img Q_0 \cap \dom S_1, \quad
      \dcal_{01}^- :=  \ker Q_0 \cap \dom S_1, \quad
      \dcal_{01} := \dcal_{01}^+ + \dcal_{01}^-.
   \end{equation}
   We say that the pair $(\lcal_0, \lcal_1)$ is \emph{in differential modular position} if
   \begin{enumerate}[(i)]
    \item \label{it:phidense} $\varphi_1 \dcal_{01} $ is dense in $\es_1$;
    \item \label{it:sortho} For all $f^\pm \in \dcal_{01}^\pm$, one has $S_1(f^+ ,f^-) = 0$. 
   \end{enumerate}
\end{definition}

This condition may seem restrictive, but it is in fact fulfilled in many relevant examples: in the conformal $U(1)$-current for half-line algebras, both in the vacuum (Example~\ref{ex:u1vac}) and in KMS states (Example~\ref{ex:u1kms}), for lightlike translated wedge algebras in the free massive field as in \cite{CLR:waveinfo}, as well as in certain abelian (Example~\ref{ex:abelian}) and finite-dimensional (Example~\ref{ex:fdstraight}) situations. Nevertheless, it is a nonempty condition (Example~\ref{ex:fdskew}). It may be seen as reminiscent of geometric modular action, as we shall see in Lemma~\ref{lemma:subsup} below.

If $(\lcal_0, \lcal_1)$ is in differential modular position, then we can define a projector $\bar Q_0$ on (a dense set of) $\es_1$ by
\begin{equation}
    \hscalar{\varphi_1 f}{\bar Q_0 \varphi_1 g}_1 = S_1( f, Q_0 g), \quad f,g \in \dcal_{01}.
\end{equation}
Because of item (\ref{it:sortho}) in the definition, this projector is actually \emph{real-orthogonal,} and hence extends uniquely to a bounded operator on all of $\es_1$.

The spaces $\dcal_{01}^\pm$ are somewhat difficult to explicitly describe in examples, we therefore give more directly applicable sufficient criteria for Def.~\ref{def:dmp}.

\begin{lemma}\label{lemma:subsup}
 Suppose that $\varphi_1 \dcal_{01}$ is dense in $\es_1$. Further suppose either
 \begin{enumerate}[(a)]
  \item \label{it:diff} $\lcal_0 \subset \lcal_1$, and the closure $i\pf \bar R_1$ of $i\pf R_1$ restricts to an $S_1$-graph-densely defined operator from $\dcal_{01}^+$ to $\lcal_0$; or  
  \item \label{it:difffact} $\lcal_0 \subset \lcal_1$, both $\lcal_0$ and $\lcal_1$ are purely factorial, and $i\pf K_1$ restricts to an $S_1$-graph-densely defined operator from $\dcal_{01}^+$ to $\lcal_0$; or  

 \item \label{it:subset} $\lcal_0 \subset \lcal_1$,  both $\lcal_0$ and $\lcal_1$ are purely factorial, and $U_1(x) \lcal_0 \subset \lcal_0$ for all $x\geq 0$ [or all $x \leq 0$]; or  

  \item \label{it:diffabel} $\lcal_0 \subset \lcal_1$, and both $\lcal_0$ and $\lcal_1$ are purely abelian; or  

  \item \label{it:supset} $\lcal_0 \supset \lcal_1$.  
 \end{enumerate}
 Then the pair $(\lcal_0, \lcal_1)$ is in differential modular position. In case (\ref{it:supset}), the associated projector $\bar Q_0$ is the identity.
\end{lemma}

Part~(\ref{it:difffact}) motivates the wording ``differential modular'', since it refers to the generator of the modular group only.

\begin{proof} 
 For (\ref{it:diff}), it suffices to show (by the assumed $S_1$-graph density) that  $S_1(f^-,f^+)=0$ for all $f^+ \in \dom \bar R_1 \cap \dcal_{01}^+$ and all $f^-\in \dcal_{01}^-$. But for these, we can write 
 \begin{equation}
   S_1(f^-,f^+) = \tau\pf(f^-, \bar R_1 f^+) = \sigma\pf(f^-, \ipf \bar R_1 f^+)  = 0,
 \end{equation}
 since $\ipf \bar R_1 f^+ \in \lcal_0$ by assumption, and $f^- \in \ker Q_0 =\lcal_0'$.

 Items (\ref{it:difffact}) and (\ref{it:diffabel}) are special cases of (\ref{it:diff}): If $\lcal_0$ and $\lcal_1$ are purely factorial, then for $f^+\in \dom K_1 \cap \dcal_{01}^+ \subset (\lcal_0)_\fact \subset (\lcal_1)_\fact$ we have $\ipf \bar R_1 f^+ = P_{1,\fact} \ipf K_1 f^+$. By assumption, $\ipf K_1 f^+ \in \lcal_0 \cap (\lcal_1)_\fact$ and hence $P_{1,\fact} \ipf K_1 f^+  = \ipf K_1 f^+ \in \lcal_0$.
 Likewise, if $\lcal_0$ and $\lcal_1$ are purely abelian, then for $f^+\in \img Q_0 = \ipf (\lcal_0)_\abel \subset \ipf (\lcal_1)_\abel$ we have $\ipf R_1 f^+ = \ipf (1-P_{1,\abel}) f^+ = \ipf f^+ \in \lcal_0$.
 
 For (\ref{it:subset}), we consider the case $x\geq 0$, the other case being analogous; that is, $U_1$ is a strongly continous semigroup on $\operatorname{img} Q_0=(\lcal_0)_\fact$ with respect to the $\kcal$-norm, and on $\dcal_{01}^+$ with respect to the $S_1$-graph norm.
 Now for $f^+ \in \dcal_{01}^+$ and $\epsilon>0$, consider
  \begin{equation}\label{eq:fplusint}
     f_\epsilon^+ = \epsilon^{-1} \int_{0}^\epsilon U_1(x) f^+\;dx \in \lcal_0;
  \end{equation}
  then $f_\epsilon^+ \to f^+$ in $S_1$-graph norm as $\epsilon \to 0$ (Lemma~\ref{lemma:graphcont}). But also $f_\epsilon^+ \in \dom K_1$ \cite[Ch.~II Lemma~1.3]{EngelNagel:semigroups} and $i\pf K_1 f_\epsilon^+ \in \lcal_0$. Since vectors of the form \eqref{eq:fplusint} are a core for the generator on $\dcal_{01}^+$, 
  we can apply part~(\ref{it:difffact}).

  For (\ref{it:supset}), we use monotonicity of the entropy (Lemma~\ref{lemma:entropygrow}) to show for $f^- \in \dcal_{01}^- \subset \lcal_0'$,
  \begin{equation}
     0 \leq S_1(f^-,f^-) \leq S_0(f^-,f^-) = 0;
  \end{equation}
  hence by the Cauchy-Schwarz inequality for $S_1$, we have $S_1(f^-,f^+)=0$ for all $f^\pm \in \dcal_{01}^\pm$.
  Thus $(\lcal_0,\lcal_1)$ is in differential modular position, but also $\idop-\bar Q_0 = 0$, i.e., $\bar Q_0 = \idop$.
\optqed{}\end{proof}

\subsection{Families of subspaces}

Closer to the applications we have in mind, we now proceed to a family of subspaces, labelled by a real parameter $t$; in particular, we are interested in the situation where the subspaces increase with the parameter.

\begin{definition}\label{def:difffamily}
  A \emph{family of differential modular inclusions} is a family  $(\lcal_t)_{t \in \rbb}$ of closed subspaces of $\kcal$ which is increasing\footnote{We do not demand that it is \emph{strictly} increasing.} (i.e., $\lcal_s \subset \lcal_t$ for each $s \leq t$) and where each pair $(\lcal_s,\lcal_t)$ ($s,t\in\rbb$) is in differential modular position (Definition \ref{def:dmp}).
\end{definition}

We will show later (Sec.~\ref{sec:halfsided}) that the usual notion of (single-particle) half-sided modular inclusions \cite{Bor:ttt,LonNotes} is a special case of Def.~\ref{def:difffamily}. However, the notion of differential modular inclusions is more general: It also applies to other situations where the modular group acts geometrically (Example~\ref{ex:u1kms}) or where the space $\lcal_t$ takes discrete steps (Example~\ref{ex:fdstraight}). Also, notice that Def.~\ref{def:difffamily} is invariant under monotonous reparametrizations of the parameter $t$, whereas half-sided modular inclusions are not.

We note that for $t\leq \hat t$, Lemma~\ref{lemma:entropygrow} gives us a canonical map $\rho_{t \hat t}: \es_{\hat t} \to \es_t$ which fulfills $\varphi_t = \rho_{t \hat t}\circ \varphi_{\hat t}$, and $\|\rho_{t \hat t}\| \leq 1$.
With respect to this inclusion map, we can now formulate some compatibility properties for the projectors $\bar Q_s$ on $\es_t$.

\begin{lemma}\label{lemma:qcompat}
   If $(\lcal_t)_{t \in \rbb}$ is a family of differential modular inclusions, then: 
   \begin{enumerate}[(a)]
    \item \label{it:sdiff} 
          For $s \leq \hat s$ and any $t$, the projectors $\bar Q_s$ and $\bar Q_{\hat s}$ on $\es_t$ fulfil $\bar Q_s \leq \bar Q_{\hat s}$.
    \item \label{it:tdiff} 
         For any $s$ and $t\leq \hat t$, let $\bar Q_s$ be the extension of $Q_s$ to $\es_t$ and $\hat Q_s$ the corresponding extension to $\es_{\hat t}$. We have $\bar Q_s \circ \rho_{t\hat t} = \rho_{t\hat t} \circ \hat Q_{s} $.
   \end{enumerate}
\end{lemma}
(Because of the last property, we will not indicate the dependence of $\bar Q_s$ on the extension space $\es_t$ beyond the proof of this lemma.)

\begin{proof}
 For (\ref{it:sdiff}), note first that $\ker Q_{s} = \lcal'_{s}$ and similarly for $\hat s$. Since $\lcal_s \subset \lcal_{\hat s}$, this yields
 \begin{equation}
    \dcal_{\hat s,t}^- \subset \dcal_{s,t}^- \, ,
    \quad \text{and }  
    (1-Q_s) f^- = f^- \quad \text{for all } f^- \in \dcal_{\hat st}^-.
 \end{equation}
 Now for any $f = f^++ f^- \in \dcal_{\hat s t}$, compute
 \begin{equation}
 \begin{aligned}
    (1-\bar Q_{\hat s})\varphi_t f &= \varphi_t (1-Q_{\hat s}) f = \varphi_t f^- = \varphi_t (1-Q_s) f^- = (1-\bar Q_s)  \varphi_t f^-
    \\&= (1-\bar Q_s) \varphi_t (1-Q_{\hat s}) f = (1-\bar Q_s)(1-\bar Q_{\hat s})\varphi_t f.
 \end{aligned}
 \end{equation}
 By density of $\varphi_t \dcal_{\hat s t}$ in $\es_t$, we obtain (using orthogonality of the projectors),
 \begin{equation}
   (1-\bar Q_s)(1-\bar Q_{\hat s}) = (1-\bar Q_{\hat s}) \quad \Rightarrow \quad Q_s \leq Q_{\hat s}. 
 \end{equation}

 Regarding (\ref{it:tdiff}): For $f\in \dcal_{s \hat t} \subset \dcal_{st}$, we compute 
 \begin{equation}
    \bar Q_s \rho_{t \hat t} \varphi_{\hat t} f = \bar Q_s \varphi_{t} f = \varphi_t Q_s f = \rho_{t \hat t} \varphi_{\hat t} Q_s f = \rho_{t \hat t} \hat Q_s \varphi_{\hat t} f.   
 \end{equation}
 By density of $\varphi_{\hat t} \dcal_{s \hat t}$ in $\es_{\hat t}$, we conclude $\bar Q_s \circ \rho_{t \hat t} = \rho_{t\hat t} \circ \hat Q_{s} $.
\optqed{}\end{proof}

\subsection{Derivatives of the entropy}

For a family of differential modular inclusions, we now investigate how the entropy $S_t(f,f)$ of a given vector $f$ depends on the parameter $t$; here we take $f\in \bar\dcal:=\cap_{t\in\rbb} \dom S_t$. Actually, in order to study the ``bulk'' vs. ``boundary'' terms mentioned in the introduction, we consider the function $T_f:\rbb^2 \to \rbb$ given by
\begin{equation}\label{deft}
  T_f(s,t) := \hscalar{ \varphi_t f}{\bar Q_s\varphi_t f}_t = \| \bar Q_s \varphi_t f\|_t^2;
\end{equation}
we have $T_f(t,t)=S_t(f)$, and we aim at estimates for $d^2 S_t/dt^2$ in terms of the partial derivatives of $T_f$, which will in general exist only in the sense of distributions.

Crucial to this analysis are certain monotonicity properties of $T_f$, in particular on the cones $\ccal_\pm:=\{ (s,t) \in \rbb^2: \pm(s-t) \geq 0 \}$.

\begin{lemma}\label{lemma:tfprop}
For any $f \in \bar\dcal$, the function $T_f$ enjoys the following properties:
\begin{enumerate}[(a)]
 \item \label{it:tmonotone} It is increasing in $t$ everywhere;
 \item \label{it:smonotone} It is increasing in $s$ everywhere, and constant in $s$ on $\ccal_+$;
 \item \label{it:diagmonotone} Along the diagonal, it is increasing, i.e., $T_f(t,t)$ is increasing in $t$.
 \item \label{it:mixed}  One has the ``mixed monotonicity'' estimate
 \begin{equation}\label{eq:mixedmon}
    \forall s<\hat s, \; t< \hat t:
    \quad
    T_f(\hat s,\hat t) - T_f(\hat s, t) - T_f(s, \hat t) + T_f(s,t) \geq 0. 
 \end{equation}
\end{enumerate}
 
\end{lemma}

\begin{proof}
 For item (\ref{it:tmonotone}), observe for $t \leq \hat t$ that 
 \begin{equation}
  T_f(s,t) = \| \bar Q_s\varphi_t  f \|_t^2 = \|  \rho_{t\hat t} \bar Q_s \varphi_{\hat t}  f \|_t^2 
 \leq \| \bar Q_s\varphi_{\hat t}  f \|_{\hat t}^2 \leq T_f(s,\hat t), 
 \end{equation}
 where Lemma~\ref{lemma:qcompat}(b) and $\|\rho_{t\hat t}\| \leq 1$ have been used.   
 Item (\ref{it:smonotone}) follows similarly from Lemma~\ref{lemma:qcompat}(a), along with $\bar Q_s = \bar Q_t=\idop$ for $s \geq t$ (Lemma~\ref{lemma:subsup}
(\ref{it:supset})).
 Item (\ref{it:diagmonotone}) is a consequence of (\ref{it:tmonotone}) and (\ref{it:smonotone}).
 For item (\ref{it:mixed}), one rewrites using Lemma~\ref{lemma:qcompat},
 \begin{equation}
    T_f( \hat s, \hat t) - T_f(\hat s, t) - T_f(s,\hat t) + T_f(s,t) 
    = \|(\bar Q_{\hat s} - \bar Q_{s}) \varphi_{\hat t} f \|^2_{\hat t} -  \| \rho_{t\hat t}(\bar Q_{\hat s} - \bar Q_{s}) \varphi_{\hat t} f \|^2_{t}  ,
 \end{equation}
 which is nonnegative since $\|\rho_{t\hat t}\|\leq 1$.
\optqed{}\end{proof}

Item (\ref{it:diagmonotone}) above implies $dS_t(f)/dt \geq 0$, at least in the sense of distributions. For $d^2S_t(f)/dt^2$, we will derive estimates stemming from item (\ref{it:mixed}).
For simplicity, let us first assume that $T_f$ is smooth outside the diagonal $s=t$, and at least $C^1$ at the diagonal. (Smoothness overall does not even occur in otherwise well-behaved examples, such as Example~\ref{ex:u1vac}).
\begin{proposition}\label{prop:secondderivsmooth}
 Let $f \in \bar\dcal$. Suppose that $T_f$ is of class $C^1$, and that there are functions $\hat T_\pm \in C^2(\rbb^2)$  such that $T_f\restrict \ccal_\pm = \hat T_\pm \restrict \ccal_\pm$. 
 Then
\begin{equation}\label{eq:secondderest}
   \lim_{\epsilon \searrow 0} \frac{\partial^2 T_f}{ \partial t^2} (t-\epsilon, t) \leq \frac{d^2 S_t(f)}{dt^2} = \lim_{\epsilon \searrow 0} \frac{\partial^2 T_f}{ \partial t^2} (t+\epsilon, t).
\end{equation}
\end{proposition}

\begin{proof}
 Since $T_f$ is $C^1$, Lemma~\ref{lemma:tfprop}(\ref{it:smonotone}) implies 
  $ \partial T_f / {\partial s} \vert_{s=t} = 0$,
 and we can differentiate this relation along the diagonal to yield
 \begin{equation}\label{eq:diagzero2}
   \frac{\partial^2 \hat T_\pm}{\partial s^2}\Big\vert_{s=t} + \frac{\partial^2 \hat T_\pm}{\partial s \partial t}\Big\vert_{s=t} = 0.
 \end{equation}
 On the other hand, $S_t=\hat T_\pm(t,t)$, which yields together with \eqref{eq:diagzero2},
 \begin{equation}
   \frac{d^2S_t}{dt^2} = \frac{\partial^2 \hat T_\pm}{\partial t^2}\Big\vert_{s=t} + \frac{\partial^2 \hat T_\pm}{\partial s \partial t}\Big\vert_{s=t}
   = \lim_{\epsilon\searrow 0} \Big(\frac{\partial^2  T_f}{\partial t^2}(t\pm\epsilon,t) + \frac{\partial^2  T_f}{\partial s\partial t}(t\pm\epsilon,t) \Big). 
 \end{equation}
 Now  $\partial^2  T_f/ \partial s\partial t$ vanishes on $\ccal_+$ by Lemma~\ref{lemma:tfprop}(\ref{it:smonotone}), and is nonnegative on $\ccal_-$ by Lemma~\ref{lemma:tfprop}(\ref{it:mixed}); hence the result follows.
\optqed{}\end{proof}

In other words, $d^2S_f/dt^2$ is bounded above und below by a ``bulk term'' (determined by the change of modular data in $R_t$), but the lower bound may allow for a positive ``boundary term'' (involving also a change in $\bar Q_s$). 

It is instructive to look at estimate \eqref{eq:secondderest} in specific examples. In certain situations, in particular in the conformal $U(1)$-current in the vacuum (Example~\ref{ex:u1vac}), one has $\partial^2 T_f / \partial t^2=0$ on $\ccal_-$, and the only contribution to $d^2S/dt^2$ is the ``boundary term'' $\partial^2 T_f/\partial s \partial t \geq 0$. Thus \eqref{eq:secondderest} implies convexity of the entropy in $t$ in this case. However, in other situation, such as thermal states on the conformal $U(1)$-current (Example~\ref{ex:u1kms}), or even when just reparametrizing a half-sided modular inclusion (Example~\ref{ex:u1rescale}), the ``bulk term'' $\partial^2 T_f / \partial t^2$ need not vanish, and indeed can take any sign. Thus $S_t(f)$ need not be convex in $t$, while the estimate \eqref{eq:secondderest} still holds.

We now want to establish a generalization of Prop.~\ref{prop:secondderivsmooth} without smoothness assumptions on $T_f$.
In preparation, we first prove:

\begin{lemma}\label{lemma:scont}
   For almost every $t \in \rbb$, the function $T_f$ is continuous at the point $(t,t)$.
\end{lemma}

\begin{proof}
  The map $t \mapsto T_f(t,t)$ is monotonic, hence continuous almost everywhere; we fix a point $t$ of continuity. 
  Consider a sequence $(s_n,t_n)$ which converges to $(t,t)$, and set $u_n:=\min\{s_n,t_n\}$, $v_n:=\max\{s_n,t_n\}$.
  Since $T_f$ is increasing in both variables by Lemma~\ref{lemma:tfprop}(\ref{it:tmonotone}),(\ref{it:smonotone}), we have
  \begin{equation}\label{eq:tsandwich}
     T_f(u_n,u_n) \leq T_f(s_n,t_n) \leq T_f(v_n,v_n).
  \end{equation}
  As $n \to \infty$, both sides of this inequality tend to $T_f(t,t)$, showing that $T_f$ is continuous (in two variables) at $(t,t)$.
\optqed{}\end{proof}

Further, we note that $S_t$ and $T_f$ are locally integrable (due to their monotonicity properties) and hence can be understood as distributions in $C_c^\infty(\rbb)'$ and $C_c^\infty(\rbb^2)'$ respectively.
Regarding test functions, we fix---for all what follows---a nonnegative function $h \in C_c^\infty(\rbb_+)$ with $\int h = 1$, and for any $g \in C_c^\infty(\rbb)$ and $\epsilon>0$, we define $g_\epsilon^\pm \in C_c^\infty(\rbb^2)$ by
\begin{equation}
   g_\epsilon^\pm (s,t) := \frac{1}{\epsilon} g(t)  h\Big(\pm\frac{s-t}{\epsilon}\Big),
\end{equation}
which has support in the interior of $\ccal_\pm$.  The dual pairing between distributions and test functions will be denoted as $\dev{ \cdotarg }{\cdotarg}$. 
With this notation, our generalisation of Prop.~\ref{prop:secondderivsmooth} to the non-smooth case is:

\begin{theorem} \label{theorem:secondderivnonsmooth}
Let $f \in \bar\dcal$. For any nonnegative $g \in C_c^\infty(\rbb)$, one has
\begin{equation}\label{eq:secondderivnonsmooth}
    \limsup_{\epsilon \searrow 0}  \dev{ \frac{\partial^2 T_f}{\partial t^2} + \frac{\partial^2 T_f}{\partial s^2}}{ g_\epsilon^- }  
     \leq \dev{ \frac{d^2S_t(f)}{dt^2}}{ g } 
     = \lim_{\epsilon \searrow 0} \dev{ \frac{\partial^2 T_f}{\partial t^2} }{ g_\epsilon^+ }.
\end{equation}
\end{theorem}

\begin{proof}
Due to local boundedness of $T_f$, we have
 \begin{equation}
\int dt\, g''(t) T_f(t,t) = \lim_{\epsilon \searrow 0} \int du\, dt\, g''(t)\frac{1}{\epsilon}h\big(\pm\frac{u}{\epsilon}\big) T_f(t+u,t)
\end{equation}
by dominated convergence together with Lemma~\ref{lemma:scont}. After a change of coordinates ($s= t+u$), this equality reads
\begin{equation}
\dev{ \frac{d^2S_t(f)}{dt^2}}{ g }  
= \lim_{\epsilon \searrow 0} \dev{ \frac{\partial^2 T_f}{\partial s^2} + \frac{\partial^2 T_f}{\partial t^2}
 + 2 \frac{\partial T_f}{\partial s\partial t}} {g_\epsilon^\pm}.
\end{equation}
Now by Lemma~\ref{lemma:tfprop}, the partial derivatives by $s$ vanish in the interior of $\ccal_+$, and $\partial T_f/ \partial s\partial t \geq 0$ in the interior of $\ccal_-$, yielding the proposed result.
\optqed{}\end{proof}

Note the extra boundary term $\partial^2 T_f / \partial s^2$ on the left-hand side of \eqref{eq:secondderivnonsmooth}, which may have any sign. This term indeed occurs in Examples~\ref{ex:fdstraight} and \ref{ex:abelianinclusion} and saturates the inequality there, hence cannot be omitted.

Thus convexity ($d^2S/dt^2\geq 0$) can fail or a number of reasons. This is already apparent from our definitions: our notion of a ``family of differential modular inclusions'' in Definition~\ref{def:difffamily} is invariant under monotonous reparametrizations of the parameter $t$, while convexity of $S_t$ is clearly not preserved under such reparametrizations in general. In fact, under mild conditions (e.g., if $S_t$ is strictly monotonous and at least continuous, but without restrictions on the second derivative), there exists a monotonous reparametrization of the family such that the resulting entropy function is convex (in fact, linear).

\section{Half-sided modular inclusions}\label{sec:halfsided}

In this section, we show that the usual notion of half-sided modular inclusions of algebras \cite{Bor:ttt}, via its analogue on the level of symplectic Hilbert spaces \cite{LonNotes}, fits into the framework of this paper; more specifically, every half-sided modular inclusion yields a family of differential modular inclusions in the sense of Def.~\ref{def:difffamily}.

To that end, we first analyze an explicit example of half-sided modular inclusions (in a sense, the smallest nontrivial one), namely, the symplectic spaces of the conformal $U(1)$-current in the vacuum state (Sec.~\ref{sec:current}). For this model, convexity of the entropy was shown to hold in \cite{Longo:entropyDist}; we show that it fits within our framework of \emph{family of differential modular inclusions}. Then, we decompose a general half-sided modular inclusion of symplectic Hilbert spaces into direct summands equivalent to the $U(1)$-current or to a trivial inclusion, thus lifting our results to the general case.

However, let us first note that our structures are indeed preserved under taking direct sums.

\begin{lemma}\label{directsum}
Let $\mathcal{I}\subseteq \mathbb{Z}_+$. For every $n\in\mathcal{I}$, let $(\kcal_n,\tau_n,\sigma_n)$ be a symplectic Hilbert space (Def.~\ref{def:space}). 
\begin{enumerate}[(a)]
\item \label{it:spacesum} $(\kcal,\tau,\sigma):=(\oplus_{n\in\mathcal{I}}\kcal_n, \oplus_{n\in\mathcal{I}}\tau_n, \oplus_{n\in\mathcal{I}}\sigma_n)$
is a symplectic Hilbert space.

\item \label{it:entrosum}
If $\lcal_n \subset \kcal_n$ are closed subspaces, and $\lcal := \oplus_{n\in\mathcal{I}}\lcal_n$, then we have
for $f=\sum_{n\in\mathcal{I}} f_n\in \kcal$,
\begin{equation}\label{add1entr}
S_{\lcal}(f,f)=\sum_{n\in\mathcal{I}} S_{\lcal_n}(f_n,f_n).
\end{equation}
\item \label{it:inclsum} 
Suppose that, for each $n\in\mathcal{I}$, $\{\lcal_t^n\}_{t\in\mathbb{R}}$, with $\lcal_t^n\subset \kcal_n$, is a family of differential modular inclusions for  $(\kcal_n,\tau_n,\sigma_n)$. Then
$$\{\lcal_t\}_{t\in\mathbb{R}}:=\{\oplus_{n\in\mathcal{I}}\lcal_t^n\}_{t\in\mathbb{R}}$$
 is a family of differential modular inclusions for $(\kcal,\tau,\sigma)$;
 and for $f=\sum_{n\in\mathcal{I}} f_n\in \cap_t\dom S_{\lcal_t}$, we have 
\begin{equation}\label{add1t}
T_f(s,t)=\sum_{n\in\mathcal{I}} T_{f_n}^{n}(s,t),
\end{equation}
where $T_{f_n}^{n}$ is the function \eqref{deft} associated with the family $\{\lcal_t^n\}$. 
\end{enumerate}
\end{lemma}
\begin{proof}
(\ref{it:spacesum}) is immediate, and (\ref{it:entrosum}) follows from the expression for the relative entropy $S_{\lcal_t}(f,f)$ in Theorem \ref{theorem:entroGen}. For (\ref{it:inclsum}), note that condition (ii) of Definition \ref{def:dmp} is implied by Eq.~\eqref{add1entr}, noting that also the projectors $Q_t$ decompose along the direct sum. 
To show that, for every $s,t\in\mathbb{R}$, condition (i) of Definition \ref{def:dmp} holds for the pair of subspaces $(\lcal_t,\lcal_s)$, let $\dcal^n_{st}$ be the subspace defined in \eqref{fundamdom} corresponding to the pair of subspaces $(\lcal_t^n,\lcal_s^n)$.
Consider
$$\bar\dcal_{st}:=\{f\in\mathcal{H}:f_n\neq0 \text{ for finitely many } n\in\mathcal{I}, f_n\in \dcal_{st}^n\}\subseteq \dcal_{st},$$
where $\dcal_{st}$ is the subspace defined in \eqref{fundamdom} relative to the pair of subspaces $(\lcal_t , \lcal_s)$ in $\kcal$.
Since the family of subspaces $\{\lcal^n_t\}_{t\in\mathbb{R}}$ is by hypothesis a differential modular inclusion, we can find for every $f\in\dom(S_{\lcal_t})$ a sequence $\{g_k\}_{k\in\mathbb{Z}_+}\subset \bar\dcal_{st}$, defined as $(g_k)_n:=0$ if $n>k$, and such that 
$$S_{\lcal_t^n}((g_k)_n-f_n,(g_k)_n-f_n)\leq \frac{1}{k}\frac{1}{2^n}$$
 for $n\leq k$.
We thus have by \eqref{add1entr}
\begin{equation}
\begin{aligned}
S_{\lcal_t} & (g_k-f,g_k-f)=\sum_{n\in\mathcal{I}}S_{\lcal_t^n}((g_k)_n-f_n,(g_k)_n-f_n)
\\
&= \sum_{0\leq n\leq k}S_{\lcal_t^n}((g_k)_n-f_n,(g_k)_n-f_n) + \sum_{n>k}S_{\lcal_t^n}(f_n,f_n)\leq \frac{1}{k}\sum_{n\geq 0}\frac{1}{2^n} +\sum_{n>k}S_{\lcal_t^n}(f_n,f_n),
\end{aligned}
\end{equation}
which converges to $0$ as $k\rightarrow\infty$ since $f$ has finite entropy.---Finally, one verifies that also the projectors $\bar Q_t$ and $\dom S_{\lcal_t}$ decompose along the direct sum, hence Eq.~\eqref{add1t} follows from  \eqref{add1entr}.
\optqed{}\end{proof}

\subsection[The U(1)-current in the vacuum]{The $U(1)$-current in the vacuum}\label{sec:current}

We consider the symplectic space for the $U(1)$-current, namely (in ``configuration space'' representation) $C_c^\infty(\rbb)$ equipped with the symplectic form
\begin{equation}\label{symplU1}
   \sigma(f,g) = \int f(x) g'(x) dx.
\end{equation}
for $f,g\in C_c^\infty(\rbb)$.

The vacuum state is the pure quasifree state induced by the bilinear form
\begin{equation}
  \tau(f,g) = \re \int_0^\infty dp\,p\,\tilde f(-p) \tilde g(p),
\end{equation}
where $\tilde f$ denotes the Fourier transform
\begin{equation}
\tilde f(p)=\frac{1}{2\pi}\int_{\mathbb{R}}e^{-ipx}f(x)dx.
\end{equation}
The closure of $C_c^\infty(\rbb)$ in the topology induced by $\tau$ is $\kcal=L^2_{\mathbb{C}}(\rbb_+,p\,dp)$. This is a complex Hilbert space with complex scalar product $\hscalar{\cdot}{\cdot}$ and indeed,
\begin{equation}\label{eq:cvacscalar}
   \hscalar{f}{g} = \int_0^\infty dp\,p\,\overline{\tilde f(p)} \tilde g(p) 
   = \tau(f,g) + i \sigma(f,g), \quad f,g\in C_c^\infty(\rbb).
\end{equation}
Thus, $(\kcal,\tau,\sigma)$ is a pure symplectic Hilbert space (see Remark \ref{remark:pure}).

Let $\mathcal{I}$ be the set of open, proper (bounded or unbounded) intervals of $\mathbb{R}$. For an interval $I\in\mathcal{I}$, let $\lcal^{U(1)}(I)$ be the closure of $C_c^\infty(I)$ in $\mathcal{K}$. The net $\{\lcal^{U(1)}(I)\}_{I\in\mathcal{I}}$ is a local net of standard and factorial\footnote{That is, $\lcal^{U(1)}(I) = \lcal^{U(1)}(I)_\fact$; see, e.g., \cite[Section 4.2]{LonNotes}.} subspaces of $\mathcal{K}$, the well known $U(1)$-current net at single-particle level (restricted to the real line $\mathbb{R}$). Its extension to the circle $S^1$ is covariant with respect to the action of the lowest weight $1$ positive energy irreducible representation of the M\"obius group, $V$.
The latter, restricted to the subgroup $\mathbf{P}$ generated by translations and dilations (denoted by $t\mapsto\vartheta(t)$ and $s\mapsto\delta(s)$ respectively), is given explicitly on $\mathcal{K}=L^2(\rbb_+,p\,dp)$ by
\begin{align}
&(V(\vartheta(t))f)(p)=e^{itp}f(p), \label{trasl}\\
&(V(\delta(s))f)(p)=e^{-2\pi s} f(e^{-2\pi s}p).\label{dil}
\end{align}
This yields the unique irreducible, strictly positive energy representation of the group $\mathbf{P}$; see, e.g., \cite[Section 6.7]{FollHarm}.

For brevity, for $t\in\mathbb{R}$, we denote by $\lcal^{U(1)}_{t}:=\lcal^{U(1)}((-\infty,t))$, $\Delta_t:=\Delta_{\lcal^{U(1)}_t}$, $K_t:=-\log(\Delta_t)$ and with mild abuse of notation we omit the identification between the configuration space representation and $L^2_{\mathbb{C}}(\rbb_+,p\, dp)$.
By the Bisognano-Wichmann theorem for M\"obius covariant local nets of standard subspaces \cite[Theorem 3.3.1]{LonNotes} we have that 
\begin{equation}\label{eqbw}
\Delta^{-is}_0=V(\delta(-2\pi s)),
\end{equation}
and $\Delta_t$ for other $t$ is then determined by translation covariance; in particular we get  (``in configuration space'')
\begin{equation}
K_t = -\log \Delta_t =  2 \pi i (x-t) \partial_x.
\end{equation}
Let $Q_t$ denote the projection \eqref{def:cutproj} relative to the subspace $\lcal_t^{U(1)}$; since the space is factorial, it acts by
\begin{equation}\label{eq:qtu1}
\begin{aligned}
Q_t: \lcal_t^{U(1)} + \lcal_t^{U(1)\prime}&\rightarrow \lcal_t^{U(1)},\\
h+h^\prime &\mapsto h.
\end{aligned}
\end{equation}

\begin{proposition}
For $f\in C^\infty_c(\mathbb{R})\subset\mathcal{K}$ we have
$$Q_t K_t f(x)=\Theta(t-x) K_t f(x)$$
for almost all $x\in\mathbb{R}$, where $\Theta$ denotes the Heaviside function.
\end{proposition}
\begin{proof}
Let $g_\pm:=\Theta(\pm(t-\cdotarg))K_t f$. Since $g_\pm$ are continuous and piecewise differentiable functions, they are elements of $\kcal$.
Also, $g_+ \in \lcal_t^{U(1)}$, since clearly $\sigma(g_+,\varphi)=0$ for all $\varphi\in \lcal_t^{U(1)\prime}$, and similarly $g_- \in  \lcal_t^{U(1)\prime}$. By \eqref{eq:qtu1}, we now have $Q_t K_t f = Q_t(g_+ + g_-) = g_+$ as claimed.
\optqed{}\end{proof}

One then immediately finds for the relative entropy by applying formula \eqref{entrform}:
\begin{proposition}
Let $f\in C^\infty_c(\mathbb{R})\subset\mathcal{K}$. We have 
\begin{equation}\label{eq:cvacentro}
 S_{\lcal^{U(1)}_t}(f,f)  = 2\pi \int_{-\infty}^t (t-x) f'(x)^2 dx.
\end{equation}
\end{proposition}

Further, the spaces fit into our framework of differential modular inclusions (Definition \ref{def:difffamily}):
\begin{proposition}\label{vacuumdiff}
$\{\lcal_t^{U(1)}\}_{t\in\mathbb{R}}$ is a family of  differential modular inclusions.
\end{proposition}
\begin{proof}
We must prove that, for every $s,t\in\mathbb{R}$, conditions (\ref{it:phidense}) and (\ref{it:sortho}) in Definition \ref{def:dmp} hold for the pair of subspaces $(\lcal_s^{U(1)},\lcal_t^{U(1)})$. In fact, once (\ref{it:phidense}) is shown, (\ref{it:sortho}) is obtained immediately from points (\ref{it:subset}) and (\ref{it:supset}) of Lemma \ref{lemma:subsup}, which apply since the modular group acts geometrically by Eq.~\eqref{eqbw}. 

Now for condition (\ref{it:phidense}) in Definition \ref{def:dmp}, note that $C^\infty_c((-\infty,t))\subset \lcal^{U(1)}_t$ is a core for the generator $\ipf K{_t\restriction{\lcal_t^{U(1)}}}$ of the $\kcal$-strongly continuous one-parameter group $s\rightarrow\Delta^{is}_{t} \restriction {\lcal^{U(1)}_t}$, as it is a dense invariant subset of $\dom \ipf K{_t\restriction {\lcal_t^{U(1)}}}$ \cite[Proposition 1.7, Chapter 2]{EngelNagel:semigroups}. Using Lemmas \ref{density64} and \ref{core}, we find that  $\varphi_{\lcal_t^{U(1)}}(C^\infty_c((-\infty,t)))$ is dense in $\es_{\lcal_t^{U(1)}}$. Thus to conclude the proof we only have to show that vectors in $C^\infty_c((-\infty,t))\subset \lcal_t^{U(1)}$ can be approximated by vectors in $\lcal_t^{U(1)}\cap\dcal_{st}$ in $\|\cdot\|_t$, where $\dcal_{st}$ is defined in \eqref{fundamdom}. Suppose $s\leq t$; the proof for $s>t$ is very similar.

Note that
\begin{equation}
\hat\dcal_{st}:=\{f\in C^\infty_c(\mathbb{R}): \supp(f)\subset (-\infty,t)\setminus\{s\}\}\subset\dcal_{st}\cap \lcal_t^{U(1)}
\end{equation}
and $\hat\dcal_{st}\subset C^\infty_c((-\infty,t))$.
Consider the map
\begin{equation}\label{eq:phiu1}
\begin{aligned}
\varphi: C^\infty_c((-\infty,t))&\rightarrow L^2((-\infty,t),(t-x)dx),\\
f&\mapsto f^\prime,
\end{aligned}
\end{equation}
which, by \eqref{eq:cvacentro}, is an isometry if $C^\infty_c((-\infty,t))$ is equipped with the norm $\|\cdot\|_t.$ To get the claim we show that the closure of $\varphi(\hat\dcal_{st})$ is the whole space $L^2((-\infty,t),(t-x)dx)$.

 To that end, we check that the orthogonal complement of $\varphi(\hat\dcal_{st})$ in $L^2((-\infty,t),(t-x)dx)$ is trivial:
\begin{equation}
0=\int_{-\infty}^t (t-x) f'(x)g(x)dx
\end{equation}
for all $f\in \hat\dcal_{st}$ implies that $g(x)=\frac{c}{t-x}$ for Lebesgue-almost all $-\infty < x < s$ with some constant $c$. But such $g$ is not in $L^2((-\infty,s),(t-x) dx)$ unless $c=0$. Thus $g$ must vanish almost everywhere in $(-\infty,s)$, and by a similar argument, also in $(s,t)$. 
\optqed{}\end{proof}

As a byproduct of the proof above, we see that the space $\es_t$ can be identifed with $L^2((-\infty,t),(t-x)dx)$ via the map \eqref{eq:phiu1}, with the projectors $\bar Q_{s}$ being multiplication with the characteristic function of $(-\infty,s)$.

Note that the inclusion
\begin{equation}
\lcal^{U(1)}((1,\infty))\subset \lcal^{U(1)}((0,\infty))
\end{equation}
is a +half-sided modular inclusion (see Definition \ref{hsm} below). It is indeed the unique nontrivial irreducible +half-sided modular inclusion up to unitary equivalence \cite[Corollary 4.3.2]{LonNotes}.
Similarly
\begin{equation}\label{U1hsm}
\lcal_{-1}^{U(1)}\subset \lcal_0^{U(1)}
\end{equation}
is the unique nontrivial irreducible $-$half-sided modular inclusion up to unitary equivalence.

\subsection{Decomposition}\label{sec:decomp}
In this section we show that the family of standard subspaces induced by a half-sided modular inclusion yields a family of differential modular inclusions (Definition \ref{def:difffamily}).
We start by recalling the notion of (single-particle) half-sided modular inclusion and some of its relevant consequences following \cite{LonNotes}.
\begin{definition}\label{hsm}
Let $K\subset H$ be real standard subspaces of a complex Hilbert space $\mathcal{H}$. If 
$$\Delta^{-it}_H K\subset K \text{ for every } \pm t\geq0,$$
the inclusion $K\subset H$ is called a $\pm$half-sided modular inclusion of standard subspaces.
\end{definition}
In our context, we will work with $-$half-sided modular inclusions only.

\noindent
As above, let $\mathbf{P}$ denote the group generated by translations and dilations on the real line $\mathbb{R}$, which we denote respectively with $\vartheta$ and $\delta$, i.e.
$\vartheta(t)(x)=x+t$, $\delta(s)(x)=e^sx$ for $t,s,x\in\mathbb{R}$. We denote by $\delta_1$ the one-parameter subgroup of $\mathbf{P}$ of dilations of the interval $(1,\infty)$, i.e. $\delta_1 (s)= \vartheta(-1)\delta(s)\vartheta(1)$.

A unitary representation $V$ of the group $\mathbf{P}$ is said to have \emph{positive energy} if the generator of the subgroup of translations, $t\rightarrow V(\vartheta(t))$, is a positive operator. It is said to be \emph{nonsingular} if the kernel of the generator of translations is trivial.

The following is the single-particle version of Wiesbrock's theorem for half-sided modular inclusions \cite{wiesbrock1993}, see \cite[Theorem 2.4.1]{LonNotes}.
\begin{theorem}\label{wiesbrock}
Let $K\subset H$ be a $-$half-sided modular inclusion of standard subspaces in a Hilbert space $\mathcal{H}$.
There exists a positive energy unitary representation $V$ of $\mathbf{P}$ on $\mathcal{H}$ determined by
$$V(\delta(2\pi s))=\Delta_H^{+is},\hspace{3mm} V(\delta_1(2\pi s))=J_H \Delta_K^{+is} J_H.$$
The translation unitaries $V(\vartheta(t))$ are defined by 
$$V(\vartheta(e^{2\pi t}-1))=\Delta_H^{+it}J_H\Delta_K^{-it}J_H$$
and satisfy $V(\vartheta(s))H\subset H$, $s\leq 0$, and $K = V(\vartheta(-1))H$.
\end{theorem}

\begin{definition}
Let  $K\subset H$ be a $-$half-sided modular inclusion and $V$ be its induced representation of $\mathbf{P}$ from Theorem \ref{wiesbrock}. $K\subset H$ is said to be 
\begin{enumerate}[(i)]
\item \emph{irreducible} if $V$ is irreducible;
\item \emph{nondegenerate} if $V$ is nonsingular;
\item \emph{trivial} if $K=H$.
\end{enumerate}
\end{definition}

The following statement is the content of \cite[Corollary 4.3.2]{LonNotes} which is obtained by decomposing the representation $V$ into irreducibles.
\begin{proposition}\label{decomposition}
Let $K\subset H$ be a $-$half-sided modular inclusion. Then it is canonically a direct sum of a nondegenerate $-$half-sided modular inclusion and a trivial $-$half-sided modular inclusion.
If  $K\subset H$ is a nondegenerate $-$half-sided modular inclusion then it is a countable direct sum of irreducible $-$half-sided modular inclusions unitary equivalent to $\lcal_{-1}^{U(1)}\subset \lcal_0^{U(1)}$ in \eqref{U1hsm}. 
\end{proposition}
From the latter decomposition, we easily derive our desired result.
\begin{proposition}
Let $K\subset H$ be a $-$half-sided modular inclusion of standard subspaces in a Hilbert space $\mathcal{H}$. Then $\{V(\vartheta(t))H\}_{t\in\mathbb{R}}$ is a family of differential modular inclusions. 
\end{proposition}
\begin{proof}
Combining Propositions \ref{decomposition} and \ref{vacuumdiff} we have that $\{V(\vartheta(t))H\}_{t\in\mathbb{R}}$ decomposes into a direct sum of families of differential modular inclusions, thus the conclusion follows by Lemma \ref{directsum}.
\optqed{}\end{proof}

Of course a similar statement can be obtained starting from a $+$half-sided modular inclusion.

\section{Examples}\label{sec:examples}

\subsection{Quantum mechanics}

As the simplest, but instructive example, let us consider a finite-dimensional symplectic Hilbert space; for concreteness, $\kcal=\cbb^n$ with $\sigma(f,g)=\im ( f, g )$, where $(f, g)$ denotes the standard scalar product on $\cbb^n$, and $\tau(f,g)=\re(f, M g )$ with some matrix $M \geq 1$. (If $M$ has no eigenvalue of 1, this may be interpreted as a thermal state on $n$ independent harmonic oscillators, cf.~\cite{NarnhoferThirring:KMS}, and $M=1$ corresponds to the ground state of the oscillators.)

\begin{example}\label{ex:fdstraight}
   Let $E$ be a $(\cdotarg,\cdotarg)$-orthogonal projector that commutes with $M$, and set $\lcal=E\kcal$. Then 
   \begin{equation}\label{eq:sfinite}
       S_\lcal(f) = 2 ( f, E \arcoth (M) f )    
   \end{equation}
where we read $\arcoth (M)$ as $\infty$ on eigenspaces of $M$ for eigenvalue 1.
   If $E_t:=E(-\infty,t)$ is the spectral family of $M$, then $\lcal_t = E_t \kcal$ is a family of differential modular inclusions with
  \begin{equation}\label{eq:tfinite}
     T_f(s,t) = 2 (f, E_s E_t \arcoth (M) f).
\end{equation}

\end{example}

\begin{proof}
   For the first statement, by choosing a basis in which both $M$ and $E$ are diagonal and applying Lemma~\ref{directsum}, it is sufficient to prove the statement for $n=1$. In this case, either $E=0$ (in which case the statement is trivial) or $E=1$ (assumed in the following). Hence $\lcal=\kcal$ and $M = m\idop$ with some $m \geq 1$. Assume first $m=1$. In that case, $(\kcal,\tau,\sigma)$ is pure, and as in Remark~\ref{remark:infinity}, one has $\lcal_\infty = \kcal \oplus 0$, $\lcal_0= 0 \oplus \kcal$. From Proposition~\ref{prop:entroInf}, one sees that both sides of \eqref{eq:sfinite} are infinite (unless $f=0$, in which case both are 0).
   
   Now let $m>1$. In that case, one has $D=-im^{-1}$ and 
   \begin{equation}
   \ipf = \begin{pmatrix}
             i m^{-1}  & i\sqrt{1-m^{-2}}
             \\
             i \sqrt{1-m^{-2}} & -i m^{-1}
         \end{pmatrix}.
\end{equation}
One then computes the modular operator $\Delta_\lcal$ of the spaces $\lcal=\kcal \oplus 0 $ and  $\ipf \lcal$ to be
\begin{equation}
   \Delta =\begin{pmatrix}
      1 & -2 (m^2-1)^{-1/2} \\
      -2 ( m^2-1)^{-1/2} &   ({m^2+3})/({m^2-1})
    \end{pmatrix},
\end{equation}
so that $\log \Delta_\lcal$ has the eigenvalues $\pm 2 \arcoth (m)$, 
and the projector $Q_\lcal$ is obtained (for example by Lemma~\ref{lemma:pfmodular}) as
\begin{equation}
   Q_\lcal = \begin{pmatrix}
           1 & \sqrt{ m^2-1} \\ 0 & 0
         \end{pmatrix}.
         \quad
\end{equation}
From Proposition~\ref{prop:entroF}, one then obtains \eqref{eq:sfinite}.

For the differential modular inclusion, again by Lemma~\ref{directsum} it suffices to consider only the case $n=1$. If $E_t=0$, then $S_t=0$ and  $(\lcal_s,\lcal_t)$ is trivially in differential modular position; both sides of \eqref{eq:tfinite} vanish. If instead $E_t=E_s=1$, i.e., $\lcal_s=\lcal_t$,  differential modular position is clear and $T_f(s,t)=S_t(f)$ in agreement wih \eqref{eq:tfinite}. Hence let $E_t=1$ and $E_s=0$. If $m=1$, then $f \in \dom S_t= \{0\}$ must vanish, and \eqref{eq:tfinite} holds trivially. If $m>1$, then $Q_s=0$, $\dcal_{st}^+ = \{0\}$, $\dcal_{st}^-=\kcal$, one has $\bar Q_s = 0$, and both sides of \eqref{eq:tfinite} vanish.
\optqed{}\end{proof}

Of course, the same results for the entropy are obtained in the usual formalism representing thermal states of the harmonic oscillator as density matrices.
It is interesting to note how our inequalities for $d^2S_t/dt^2$ work out in this example. Writing $M = \sum_j m_j E_j$ in spectral decomposition,
the terms in Theorem~\ref{theorem:secondderivnonsmooth} are (in suggestive notation)
\begin{equation}
  \frac{\partial^2 T_f}{\partial s^2} (t-0,t) =  \frac{d^2 S_t(f)}{dt^2} =   2\sum_j \delta'(t-m_j) \arcoth(m_j) \, (f, E_j f),
 \quad  \frac{\partial^2 T_f}{\partial t^2} (t-0,t) = 0, 
\end{equation}
so that the inequality \eqref{eq:secondderivnonsmooth} turns into an equality; note that the distribution $\delta'$ does not have a definite sign, hence finding nontrivial lower estimates would not be possible.

In the above example, our condition of \emph{differential modular position} was satisfied because $M$ leaves each subspace $\lcal_t$ invariant. Clearly, this will not be true for more general subspaces. We provide an explicit counterexample:

\begin{example}\label{ex:fdskew}
 In the above setting, let $n=2$, $M=\operatorname{diag}(2,3)$, let $\lcal_0\subset\kcal$ be the subspace spanned over $\rbb$ by the vectors $(1,1)$ and $(i,0)$, and let $\lcal_1=\kcal$. Then $(\lcal_0,\lcal_1)$ is \emph{not} in differential modular position.
\end{example}
\begin{proof}
  By explicit matrix computation, one can find the modular operators related to $\lcal_0$ and $\lcal_1$, and hence an explicit formula for $Q_0$ and $S_1$. For $f \in \kcal$, consider
 \begin{equation}
      \delta := S_1(f,f) - S_1(Q_0 f, Q_0 f); 
\end{equation}
  one finds that if $f = (a,b)$ with $a,b \in \rbb$, then
  \begin{equation}
     \delta = (b^2-a^2) \log 2.
  \end{equation}
  This is in general not positive, hence $Q_0$ cannot be orthogonal with respect to $S_1$.
\optqed{}\end{proof}
In the same situation, one finds that $S_1(f,f) - S_1(f, Q_0 f) = b(b-a)\log 2$; hence $T_f(s,t)$ defined as in \eqref{deft} may be decreasing in $s$, and the conclusion of Lemma~\ref{lemma:tfprop} fails.

\subsection[Conformal U(1)-current, vacuum state]{Conformal $U(1)$-current, vacuum state}

As an example from quantum field theory, we can take the conformal $U(1)$-current in the vacuum state, as already treated in Sec.~\ref{sec:current}. We briefly summarize the results here, and comment how the estimates on derivatives of the entropy work out in this case.

\begin{example} \label{ex:u1vac}
 Let $(\kcal,\tau,\sigma)$ be the pure symplectic Hilbert space defined in Sec.~\ref{sec:current} and let $\lcal_t = \lcal_t^{U(1)} \subset \kcal$. 
 Then $(\lcal_t)_{t \in \rbb}$ is a family of differential modular inclusions, the space $\es_t$ is isomorphic to $L^2((-\infty,t), (t-x) dx)$ via $\varphi_t: f \mapsto f'$, and we have for $f \in C_c^\infty(\rbb)\subset \kcal$,
 \begin{align}\label{vacresult}
 T_f(s,t)  &= 2\pi \int_{-\infty}^{\min\{s,t\}} (t-x) f'(x)^2 \,dx,
&
S_t(f,f)  &= 2\pi \int_{-\infty}^t (t-x) f'(x)^2 dx.
\end{align}
\end{example}

We note here that $T_f$ is $C^1$, as well as the restriction of smooth functions to the cones $\ccal_\pm$; one has
\begin{equation}\label{vacder}
   \frac{d^2}{dt^2} S_t(f) =
   \frac{\partial^2}{\partial s \partial t} T_f\Big\vert_{s=t-0} = 2\pi f'(t)^2 \geq 0,
  \quad
  \frac{\partial^2}{\partial t^2} T_f\Big\vert_{s=t-0} = 0, 
\end{equation}
so that the second derivative of $S_t$ is positive and given by a boundary term only.

By the results of Sec.~\ref{sec:halfsided}, a similar behaviour is exhibited by general half-sided modular inclusions. Namely, in the notation of Section \ref{sec:decomp}, let $K\subset H$ be a $-$half-sided modular inclusion in the complex Hilbert space $\mathcal{H}$, $K_t:=V(\vartheta(t))H$ and $f\in\mathcal{H}$. By using Lemma \ref{directsum}, we have that $S_t(f,f)$ and $T_f(s,t)$ decompose along the direct sum provided by Proposition \ref{decomposition}. Precisely, if $f=\sum_{n>0} f_n +f_{0}$ is the corresponding decomposition of $f$, with $f_0$ being the component relative to the trivial modular inclusion,
\begin{equation*}
S_t(f,f)=\sum_{n\geq0}S_t(f_n,f_n),\hspace{3mm}T_f(s,t)=\sum_{n\geq0}T_{f_n}(s,t),
\end{equation*}
where $S_t(f_n,f_n)$ and $T_{f_n}(s,t)$ with $n>0$ are given by \eqref{vacresult}. The behaviour of the derivatives of $S_t(f,f)$, at least when taking the $f_n$ to have suitably fast decay, is analogous to \eqref{vacder}, noticing also that $S_t(f_0,f_0)$, i.e., the contribution to the relative entropy given by the trivial half-sided inclusion component, is constant in $t$. As a particular example, this applies to the subspaces associated with lightlike shifted wedges in the real scalar free field, as in \cite{CLR:waveinfo}.

The above situation is compatible with Proposition~\ref{prop:secondderivsmooth}; however, the vanishing of $\frac{\partial^2}{\partial t^2} T_f$ is clearly a special feature of this particular family of subspaces. Even a reparametrization will remove it:

\begin{example}\label{ex:u1rescale}
 In Example~\ref{ex:u1vac}, consider instead $\lcal_t = \lcal_{h(t)}^{U(1)}$ where $h$ is a smooth, strictly increasing function. 
 Then $(\lcal_t)_{t \in \rbb}$ is still a family of differential modular inclusions, but 
\begin{align}
   \frac{d^2}{dt^2} S_t(f) &= \underbrace{2\pi h'(t)^2 f'(t)^2 }_{\frac{\partial^2 T_f}{\partial s \partial t} (t-0,t)} 
    + \underbrace{2\pi h''(t) \int_{-\infty}^t f'(x)^2 \,dx }_{ \frac{\partial^2 T_f}{\partial t^2} (t-0,0)}. 
\end{align}
\end{example}

Of course, this is still compatible with Proposition~\ref{prop:secondderivsmooth}, but while the first term (the ``boundary term'') is still positive, the second derivative of $S_t$ will not be positive in general.

\subsection[Conformal U(1)-current, thermal states]{Conformal $U(1)$-current, thermal states}

Let us now consider thermal (KMS) states on the conformal $U(1)$-current, as described in \cite{BorYng:thermal}. This examples illustrates, in particular, that different scalar products $\tau$ can be chosen with respect to the same symplectic form $\sigma$, and that this leads to different relative entropies. 

Specifically, fixing $\beta>0$, we choose on the non-completed space $C_c^\infty(\rbb)$ the bilinear form
\begin{align}\label{taub}
\tau_\beta (f,g) :=\operatorname{Re} \int_0^\infty \frac{\tilde f(-p)\tilde g(p)}{1-e^{-\beta p}}p \, dp.
\end{align} 
The associated quasifree state fulfills the KMS condition with respect to translations \cite{BorYng:thermal}.

The completion of the symplectic space $C^\infty_c(\mathbb{R})$ with respect to $\tau_\beta$ is $\kcal_{\beta} := L^2_{\mathbb{C}}(\rbb_+, \frac{p}{1-\exp(-\beta p)}\,dp)$ as a real vector space; as before, we do not always denote the Forier transform explicitly. We apply the purification procedure described in Section \ref{sec:ops}:
it is easy to see that
\begin{equation}\label{tds}
\sigma(f,g) = \tau_\beta (f, Dg), \quad f,g \in \kcal_{\beta},
\end{equation}
where $D$ is the multiplication operator by $i\,(1-e^{-\beta  p })$ on $\kcal_{\beta}$.  In the polar decomposition $D= C \lvert D \rvert$, the operators $\lvert D \rvert$ and $C$ act by multiplication with $(1-e^{-\beta  p })$ and $i$, respectively. This induces a complex structure $\ipf$ on $\kcal^{\oplus}_\beta:=\mathcal{K}_{\beta}\oplus\mathcal{K}_{\beta}$ by \eqref{complexstr}; let us denote the complex scalar product as $\hscalar{\cdot}{\cdot}_\beta$.

For $t\in\mathbb{R}$, we define $\lcal^{\beta}_t:=\overline{C_c^\infty((-\infty,t)))} \subset\mathcal{K}_\beta$ where the bar indicates norm closure in $\kcal_{\beta}$.
The bilinear form $\tau_\beta$ is translation invariant, thus we have the following immediate proposition.

\begin{proposition}\label{translkms}
The symplectic action of translations $(\vartheta(t)^\star f)(x)=f(x-t)$, $f\in C^\infty_c(\mathbb{R})$, extends to a unitary representation of the group of translations $\vartheta$ on $\mathcal{K}_\beta^{\oplus}$ by
\begin{equation}
V(\vartheta(t))\big( (f\oplus 0)+\ipf(g\oplus 0) \big)=\vartheta(t)^\star f\oplus 0 + \ipf (\vartheta(t)^\star g\oplus 0), 
\end{equation}
$f,g\in C^\infty_c(\mathbb{R}).$ Furthermore it acts covariantly on the subspaces $\lcal^{\beta}_t$, namely
\begin{equation}
V(\vartheta(t))\lcal^{\beta}_s=\lcal^{\beta}_{s+t} 
\end{equation}
for every $s,t\in\mathbb{R}$.
\end{proposition}

Regarding the decomposion of $\kcal_\beta\pf$ with respect to $\lcal_t^\beta$, we find:

\begin{proposition}
For each $t$, the subspace $\lcal_t^{\beta} \oplus 0\subset\mathcal{K}^{\oplus}_\beta$  is standard and factorial, i.e., in the notation of Section~\ref{sec:ops} we have $\lcal_t^{\beta}=(\lcal^{\beta}_{t })_\fact$ and $\mathcal{K}^{\oplus}_\beta=(\lcal^{\beta }_{t})^\oplus_\fact$.
\end{proposition} 

\begin{proof}
We first show $(\lcal_t^{\beta})_\infty\equiv  \lcal_t^{\beta} \cap \ipf \lcal_t^{\beta} = \{ 0\}$. 
A generic element in this intersection is of the form $f \oplus 0 = \ipf(g \oplus 0)$ for some $f,g \in \lcal_t^{\beta}$. Since $\ipf(g\oplus 0) = -Dg \oplus C\sqrt{1 + D^2}g$, we have $C\sqrt{1+D^2}g=0$, which implies that $g=0$, since $C\sqrt{1 + D^2}$ has trivial kernel.

For establishing $(\lcal_t^\beta)_0 \equiv \lcal_t^{\beta\perp} =\{0\}$,  where $\perp$ indicates the orthogonal complement in $\kcal_\beta^{\oplus}$, we use a modified Reeh-Schlieder argument as follows. 
Let $f \in  \lcal_t^{\beta\perp}$, and let $\varphi = \varphi_a \oplus 0 + \ipf (\varphi_b \oplus 0)$ with arbitrary $\varphi_{a,b}\in C^\infty_c(\mathbb{R})$.
Denoting $\varphi_s=V(\vartheta(s))(\varphi)$ (see Proposition \ref{translkms}), we notice that $\langle f, \varphi_s \rangle_{\beta}$ is analytic in $s \in \mathbb{R}+i\mathbb{R}_+$ (indeed, we can differentiate under the integral sign by use of dominated convergence). But for large negative $s$, we have $\varphi_s \in \lcal_t^{\beta} +\ipf \lcal_t^{\beta}$, and hence $\langle f, \varphi_s \rangle_{\beta} = 0$. Due to analyticity, the same then holds for all $s$, in particular $\langle f, \varphi \rangle_{\beta}=0$. Therefore, writing $f=f_1 \oplus f_2$,
\begin{equation}
0 = \re \langle f, \varphi \rangle_\beta = \tau_{\beta}(f_1, \varphi_a)  -\sigma(f_1,\varphi_b) - \tau_{\beta} (C\sqrt{1+D^2} f_2, \varphi_b).
\end{equation}
In particular, choosing $\varphi_b=0$ and varying $\varphi_a\in C_c^\infty(\rbb)$ yields $f_1=0$; then also $f_2=0$ since $C\sqrt{1+D^2} $ has trivial kernel. Hence $f=0$.

Finally, we show $(\lcal_t^{\beta})_\abel \equiv \lcal_t^{\beta} \cap \lcal_t^{\beta\prime}=\{0\}$. 
Any $f \in \lcal_t^{\beta}$ is of the form $f = \operatorname{\tau_\beta-lim}_{n \to \infty} \varphi_n$ for some $\varphi_n \in C_c^\infty((-\infty,t))$. Since the $\tau_\beta$-norm is stronger than the $L^2$-norm, this implies $f\in L^2(\mathbb{R})$ and $\operatorname{supp} f \subset (-\infty,t]$. If now also $f \in \lcal^{\beta \prime}_t$, one has $\sigma(f,\varphi) =0$ for all $\varphi \in C_c^\infty((-\infty,t))$. This means $\int f(x)\varphi'(x)\, dx =0$ for all these $\varphi$, which by standard argument implies that $f$ is constant on $(-\infty,t)$. Hence $f=0$.
\optqed{}\end{proof}

Following \cite{BorYng:thermal}, one can determine the modular group associated with $\lcal^{\beta}_0$:

\begin{proposition}
The modular group $U_0(u):=\Delta_{\lcal^{\beta}_0}^{iu}$ of $\lcal^{\beta}_0$ as a standard subspace w.r.t.\ $\kcal^{\oplus}_{\beta}$ is given by
\begin{equation}\label{deltapartial}
U_0(u)\big(\psi \oplus 0 + \ipf(\varphi \oplus 0) \big) = \delta_u \psi \oplus 0 + \ipf(\delta_u \varphi \oplus 0)
\end{equation}
for $u\in\mathbb{R}, \varphi,\psi\in C^{\infty}_c((-\infty,0))$, and 
    \begin{equation}
\delta_u f(x) := f\Big(-\frac{\beta}{2\pi}\log\big( 1+ e^{2\pi u}( e^{-2\pi x/\beta} -1)\big)  \Big)
\end{equation}
for $f\in C^{\infty}_c((-\infty,0))$.
\end{proposition}

\begin{proof}
%
$\delta_u$ fulfils the following properties \cite[p.~620]{BorYng:thermal}:
\begin{enumerate}[(i)]
\item \label{i} It preserves the subspace $C_c^\infty((-\infty,0))$;
 \item\label{ii} it satisfies the group property, i.e., $\delta_u \circ \delta_{u'} = \delta_{u+u'}$;
 \item \label{iii} there holds $\tau_\beta(\delta_u f, \delta_u g) = \tau_\beta (f,g)$ and  $\sigma(\delta_u f, \delta_u g) = \sigma(f,g)$;
\item \label{iv} setting $\omega_2(f,g):=\tau_{\beta}(f, g)+i\sigma(f, g)$, for every $f,g \in C_c^\infty((-\infty,0))$ the function $u \mapsto \omega_2(f, \delta_u g)$ has an analytic continuation into the strip $S(-1,0)$ and
\begin{equation}
\omega_2(f,\delta_{u-i} g) = \omega_2(\delta_u g, f).
\end{equation}
\end{enumerate}
Thus item (i) and item (iii) show that \eqref{deltapartial} is a well-defined and unitary mapping, and by item (ii) $U_0$ is a unitary group. Using complex linearity of the scalar product $\hscalar{\cdot}{ \cdot}_\beta$, item (iv) implies the KMS property in the form
\begin{equation}
\hscalar{f}{U_0(u-i) g}_\beta =\hscalar{U_0(u)g}{ f}_\beta \quad \text{for all}\;\; f,g  \in \lcal_0^{\beta}, 
\end{equation}
which characterizes $U_0$ as the modular group of $\lcal_0^{\beta}$.
\optqed{}\end{proof}

We can thus compute $-\log \Delta_0$, the generator of the modular group action of $\lcal_0^{\beta}$:
for $f\in C^\infty_c((0,\infty)) \subset \lcal_0^{\beta}$ we have
\begin{equation}
 \begin{aligned}
(-\log \Delta_0)f\oplus 0 &:= \ipf\frac{d}{du} U_0(u)(f \oplus 0)\big\vert_{u=0} 
=\ipf\frac{d}{du} \delta_u f \oplus 0\big\vert_{u=0} \\
&=\ipf\frac{d}{du} f\Big(- \frac{\beta}{2\pi}\log \big( 1+ e^{2\pi u}(e^{-2\pi \cdot/\beta} -1) \big) \Big)\Big\vert_{u=0} \oplus 0 \\
&= -\ipf \big( f'(\cdot) \beta (1- e^{2\pi \cdot/\beta}) \oplus 0 \big).  
 \end{aligned}
\end{equation}
As an immediate consequence of Proposition~\ref{prop:entroF}, we now obtain:

\begin{proposition}\label{kmsentr1}
For $f\in (C^\infty_c((-\infty,t))\oplus 0)\subset \lcal^{\beta}_t$ we have 
\begin{equation}
S_{\lcal_t^{\beta}}(f,f) =  \int_{-\infty}^t (f'(x))^2 \beta (1- e^{2\pi (x-t)/\beta})\, dx.
\end{equation}
\end{proposition}

Further, analogous to the vacuum case in Sec.~\ref{sec:current}, we can show:

\begin{proposition}
$\{\lcal_t^{\beta}\}_{t\in\mathbb{R}}$ is a family of modular differential inclusions (Definition \ref{def:difffamily}).
\end{proposition}
\begin{proof}
The proof of condition (i) is very similar to Proposition \ref{vacuumdiff}. Condition (ii) follows immediately by Lemma \ref{lemma:subsup}(\ref{it:subset}).
\optqed{}\end{proof}

The entropy for general $f\in C^\infty_c(\mathbb{R})$ is less directly accessible, since the action of the modular group on $f\in C^\infty_c((t,\infty))$ is not explicitly known. Nevertheless, we can show:

\begin{proposition}
For $f\in C^\infty_c(\mathbb{R})\oplus 0$,
\begin{equation}
S_{\lcal^{\beta}_t}(f,f)=\int_{-\infty}^t dx (f^\prime(x))^2  \beta (1- e^{2\pi (x-t)/\beta})\, dx.
\end{equation}
\end{proposition}
\begin{proof}
Similarly to the proof of Proposition  \ref{vacuumdiff}, given $f\in C^\infty_c(\mathbb{R})$ and chosen $t\leq s\in\mathbb{R}$ such that $\supp(f)\subset (-\infty,s)$, we have a sequence $\{f_n\}_{n\geq0}$ in 
\begin{equation}
\hat\dcal_{ts}:=\{g\in C^\infty_c(\mathbb{R}):\supp(g)\subset(-\infty,s)\setminus \{t\}\}\subset \dcal_{ts}
\end{equation}
such that $\|f-f_n\|_s\rightarrow 0$ and thus $\|f-f_n\|_t\rightarrow 0$ since $s\geq t$. If $Q_t$ denotes the projection \eqref{def:cutproj} of the subspace $\lcal_t^\beta$, we further have that $\|f-Q_t f_n\|_t\leq \|f-f_n\|_t+ \|f_n-Q_t f_n\|_t=\|f-f_n\|_t\rightarrow 0$, thus
$S_{\lcal^{\beta}_t}(f,f)=\|f\|_t^2=\lim_{n}\|Q_t f_n\|_t^2$. If $\varphi_s$ is the extension to $\es_s$ of the isometric mapping
\begin{equation}
\begin{aligned}
\varphi_s:C^\infty_c((-\infty,s))&\rightarrow L^2((-\infty,s), \beta (1- e^{2\pi (x-s)/\beta})\, dx)\\
f&\mapsto f^\prime,
\end{aligned} 
\end{equation}
(where $C^\infty_c((-\infty,s))$ is equipped with norm $\|\cdot\|_s$), then $\varphi_s(Q_t f_n)=\Theta(t-\cdotarg)\varphi_s(f_n)$.
Multiplication by $\Theta(t-\cdotarg)$ is continuous in $L^2((-\infty,s), \beta (1- e^{2\pi (x-s)/\beta})\, dx)$, thus $\Theta(t-\cdotarg)\varphi_s(f_n)\rightarrow \Theta(t-\cdotarg)\varphi_s(f)$. This implies that
\begin{equation}
\Theta(t-\cdotarg)f_n^{\prime}=\Theta(t-\cdotarg) \varphi_s(f_n)\rightarrow \Theta(t-\cdotarg)\varphi_s(f)=\Theta(t-\cdotarg)f^{\prime}
\end{equation} 
 in $L^2((-\infty,t), \beta (1- e^{2\pi (x-t)/\beta})\, dx)$, which yields the statement using Proposition \ref{kmsentr1}.
\optqed{}\end{proof}

Let us summarize the results:

\begin{example} \label{ex:u1kms}
In the symplectic Hilbert space $(\kcal_\beta,\tau_\beta,\sigma)$ of the conformal $U(1)$-current with KMS state induced by $\tau_\beta$, consider $\lcal_t^\beta = \overline{C_c^\infty(-\infty,t)}$. Then $\{\lcal_t^\beta\}_{t\in\rbb}$ is a family of differential modular inclusions, and for  $f\in C^\infty_c(\mathbb{R})$, 
\begin{equation}
 T_f(s,t)  = \| \bar Q_s \varphi_t f\|_t^2= \int_{-\infty}^{\min\{s,t\}} dx (f^\prime(x))^2  \beta (1- e^{2\pi (x-t)/\beta})\, dx.
\end{equation}
\end{example}

Computing the second derivative of the relative entropy, we have
\begin{equation}
\label{kmsexplentr}
\frac{d^2 S_{\lcal_t^{\beta}}(f)}{d t^2} =
\underbrace{2\pi (f'(t))^2}_{  \frac{\partial^2 T_f}{\partial s \partial t} (t-0,0) \geq 0}
\underbrace{-\frac{(2\pi)^2}{\beta} \int_{-\infty}^t dx\, (f'(x))^2 e^{-2\pi (x-t)/\beta} }_{ \frac{\partial^2 T_f}{ \partial t^2} (t-0,0) \leq 0}. 
\end{equation}
Thus in the present case, the estimate in Proposition \ref{prop:secondderivsmooth} bounds the second derivative of the entropy from below by its negative bulk term, i.e., the second term in \eqref{kmsexplentr}.

\subsection{Commutative algebras}

It is instructive to consider also the case of abelian CCR algebras, i.e., subspaces $\lcal$ with $\lcal\subset\lcal'$.
\begin{example}\label{ex:abelian}
Let $(X,\mcal,\mu)$ be a measure space such that $\kcal:=L^2_\rbb(X,d\mu)$ is separable. For measurable subsets $Y\subset X$, set $\lcal_Y := L^2_\rbb (Y) \subset \kcal$. For $(\kcal,\tau,0)$ with $\tau$ the $L^2$ scalar product, identify $\kcal\pf$ with $L^2_\cbb(X,d\mu)$ as in Remark~\ref{remark:abelian}. Then $\es_Y$ can be identified with $L^2_\rbb(Y,d\mu)$ via $\varphi_Y: f \mapsto \im f \restrict Y$, and 
\begin{equation}\label{eq:l2entro}
  S_Y(f,f) = 2 \int_Y (\im f)^2 d\mu. 
\end{equation}
For any two such subsets $Y,Z$, the pair $(\lcal_Z,\lcal_Y)$ is in differential modular position, and $\bar Q_{Z}$ acts on $\es_Y$ by multiplication with the characteristic function of $Y \cap Z$.
\end{example}

\begin{proof}
  By Remark~\ref{remark:abelian}, we have $\lcal_0\pf = L^2_\cbb(Y^c,d\mu)$,  $\lcal_\abel= \lcal_Y$, $\lcal_\fact=\lcal_\infty = 0$.
  The formula \eqref{eq:l2entro}, and with it the proposed form of $\varphi_Y$, then follows directly from Proposition~\ref{prop:entroA}; note that $S_Y$ is bounded and defined on all of $\kcal\pf$.
  Also, noting that the projector $Q_Z$ (which acts by $Q_Z f = \chi_Z \im f$) is already orthogonal, we have $\dcal_{ZY}^+ = i L^2_\rbb(Z \cap Y)$,
  $\dcal_{ZY}^- = \{f \in \kcal: \im f \restrict Z = 0\}$, and one sees that $\bar Q_Z$ multiplies with $\chi_{Z \cap Y}$ in $\es_Y$.
\optqed{}\end{proof}

As a special case, let us consider:
\begin{example}\label{ex:abelianinclusion}
Let $\kcal=L^2_\rbb(\rbb)$, $\kcal\pf=L^2_\cbb(\rbb)$ with subspaces $\lcal_t= L^2_\rbb(-\infty,t)$. Then $\lcal_t$ is a family of differential modular inclusions, and we have
  $S_t(f,f) = 2\int_{-\infty}^t (\im f(x))^2 dx$ and $T_f(s,t) = 2\int_{-\infty}^{\min(s,t)} (\im f(x))^2  dx$.
\end{example}
We note that in this example, the function $T_f$ is \emph{not} $C^1$. Clearly
\begin{equation}
   \frac{d^2S_t}{dt^2} = 2\frac{d}{dt}(\im f(t))^2 ,
\end{equation}
thus $S_t$ will not be convex in general. Note that $\partial^2 T_f/ \partial^2 s = 2\frac{d}{ds}(\im f(s))^2$ and $\partial^2 T_f/ \partial t^2=0$ for $s<t$, so that the estimate in Theorem~\ref{theorem:secondderivnonsmooth} is saturated.

\section{Conclusions}\label{sec:conclusion}

In this paper, we have analyzed the relative entropy between coherent excitations of a general quasifree state on a CCR algebra, with respect to the algebra generated by a generic closed subspace. We gave an explicit description of the relative entropy in terms of single-particle modular data.

Also, we analyzed the change of the relative entropy along an increasing one-parameter family of subspaces, establishing an abstract notion of bulk and boundary changes. Convexity of the entropy (or the QNEC) is in general replaced by certain lower estimates of the second derivative, where both bulk and boundary terms can contribute.

An instrumental part of this analysis was the notion of \emph{differential modular position} of two subspaces, meaning that the projector onto one subspace is orthogonal with respect to the scalar product induced by the entropy form of the other. While this is a nontrivial condition, we showed that it is fulfilled in a number of relevant examples; in particular it includes, but generalizes, the well-known notion of half-sided modular inclusions.

As the condition of differential modular position seems a fruitful tool, it would certainly be of interest to investigate whether it holds, possibly in a generalization, in a wider context than discussed here, both in other models of (linear) quantum fields and with respect to more general positions of subalgebras than treated in examples here. In particular, one would expect that it can be formulated employing notions of category theory, akin to the ``locally covariant'' setting of quantum field theory \cite{BFV03}. We hope to report on this issue elsewhere.

Also, it would be of interest to generalize our framework beyond CCR algebras to general inclusions of von Neumann algebras; in the context of quantum field theory, this would correspond to models beyond linear fields. Clearly, a challenge is the limited availability of concrete examples beyond CCR algebras, in particular with sufficiently explicit descriptions of the relative modular operator. Possibly integrable models in low space-time dimensions, which are (fully or partially) known to fulfill quantum inequalities \cite{BostelmannCadamuroFewster:ising,BostelmannCadamuro:oneparticle}, can provide some test cases in this respect.

\appendix

\section[Relative entropy on C* and von Neumann algebras]{Relative entropy on $C^\ast$ and von Neumann algebras}\label{app:entro}

The notion of relative entropy for states on general von Neumann algebras was first introduced by Araki \cite{Araki:relent1,Araki:relent2}. We recall its definition and relevant properties, following \cite{OP04}.

Let $\mcal$ be a von Neumann algebra on a Hilbert space $\hcal$, let $\omega = \hscalar{\xi}{\cdotarg \xi}$ a vector state (with some $\xi \in \hcal$), and $\varphi$ another state on $\mcal$. 
The relative entropy between $\omega$ and $\varphi$ (with respect to $\mcal$) is defined as
\begin{equation}\label{eq:entropydef}
   \relent{\omega}{\varphi}{\mcal} = \begin{cases}
                                            -\hscalar{\xi}{ \log \Delta(\varphi/\omega_{\xi}') \xi} \quad &\text{if }\xi \in \supp \varphi,
                                            \\
                                            \infty \quad & \text{otherwise}.
                                        \end{cases}
\end{equation}
Here $\omega_{\xi}'$ is the state $\hscalar{\xi}{\cdotarg \xi}$ restricted to $\mcal'$, and $\Delta(\varphi/\omega_{\xi}')$ denotes the spatial derivative.

In the case where both $\omega$ and $\varphi$ are given by cyclic and separating vectors $\xi,\psi$, the relative modular $\Delta_{\psi,\xi}$ is defined and we have (see \cite[Theorem 5.7]{OP04}, \cite[Proposition 4.1]{CLR:waveinfo})
\begin{equation}\label{eq:entropymodgrp}
   \relent{\omega}{\varphi}{\mcal} = i \frac{d}{dt} \hscalar{\xi}{ \Delta_{\psi,\xi}^{it} \xi} \Big\vert_{t=0}.
\end{equation}

If $\A$ is a $C^\ast$-algebra and $\omega,\varphi$ are positive linear functionals on $\A$, then  $\relent{\omega}{\varphi}{\A}$ is defined as
$$\relent{\omega}{\varphi}{\A}:=\relent{\bar\omega}{\bar\varphi}{\A^{**}},$$
where the right-hand-side denotes the relative entropy with respect to the universal enveloping von Neumann algebra $\A^{**}$ of $\A$ and $\bar\omega$,  $\bar\varphi$ are the normal extensions of $\omega$, $\varphi$ to $\A^{**}$.

Suppose there is a representation $\pi$ of $\A$, $\pi:\A\rightarrow B(\mathcal{H})$, where $\omega$ is a vector state, i.e., there is $\xi\in\mathcal{H}$ with
$$\tilde\omega(\pi(a)):=\hscalar{\xi}{\pi(a)\xi}=\omega(a),\hspace{3mm}a\in\A,$$
  and for which there is a normal state $\tilde{\varphi}$ on $\pi(\A)^{\prime\prime}$ such that
$$\varphi(a)=\tilde{\varphi}(\pi(a)),\hspace{3mm}a\in\A.$$
Then by applying Kosaki's formula for the relative entropy \cite[Theorem 5.11]{OP04}, we have
\begin{equation}
\relent{\omega}{\varphi}{\A}=\relent{\tilde\omega}{\tilde\varphi}{\pi(\A)^{\prime\prime}}.
\end{equation}

We recall the following properties of the relative entropy:

\begin{lemma}\label{lemma:entrounion}\cite[Corollary 5.12, iv]{OP04}
 Let $\mathcal{M}_i$ be an increasing net of von Neumann subalgebras of $\mathcal{M}$ with the property $(\cup_i \mathcal{M}_i)^{\prime\prime}=\mathcal{M}$. Then $S_{\mathcal{M}_i}(\omega_1\restriction {\mathcal{M}_i},\omega_2\restriction {\mathcal{M}_i})$ converges to $S_{\mathcal{M}}(\omega_1,\omega_2)$, where $\omega_1,\omega_2$ are two positive normal linear functionals on $\mathcal{M}$.
\end{lemma}

\begin{lemma}\label{lemma:additivity}\cite[follows from Corollary 5.20]{OP04}
Let $\mathcal{M}_1$ and $\mathcal{M}_2$ be von Neumann algebras, let $\omega_{1},\varphi_1$ be  normal states on $\mathcal{M}_1$ and let $\omega_{2},\varphi_2$ be  normal states on $\mathcal{M}_2$. Then
\begin{equation}
\relent{\omega_{1}\otimes\omega_2}{\varphi_1\otimes\varphi_2}{\mathcal{M}_1\otimes\mathcal{M}_2}=\relent{\omega_{1}}{\varphi_1}{\mathcal{M}_1}+\relent{\omega_{2}}{\varphi_2}{\mathcal{M}_2}.
\end{equation}
\end{lemma}

\begin{lemma}\cite[follows from Theorem 5.3]{OP04}\label{lem:restrentropy}
Let $\omega$ and $\varphi$ be two normal states on a von Neumann algebra $\mathcal{M}$, and denote by $\omega_1$ and $\varphi_1$ the restrictions of $\omega$ and $\varphi$ to a von Neumann subalgebra $\mathcal{M}_1 \subset\mathcal{M}$ respectively. Then $\relent{\omega_{1}}{\varphi_1}{\mathcal{M}_1} \leq \relent{\omega}{\varphi}{\mathcal{M}}$.
\end{lemma}

\section*{Acknowledgments}

We would like to thank Christopher J.\ Fewster and Ko Sanders for helpful comments on the draft.
D.C. and S.D. are supported by the Deutsche Forschungsgemeinschaft (DFG) within the Emmy Noether grant CA1850/1-1. H.B.~would like to thank the Institute for Theoretical Physics at the University of Leipzig for hospitality.

%
%

\bibliographystyle{alpha}
\bibliography{qei}

\end{document}